\documentclass[preprint,a4paper,10pt]{article}

\usepackage{amssymb}
\usepackage{amsfonts}
\usepackage{graphicx}
\usepackage{amsmath}
\usepackage{latexsym}
\usepackage{amssymb}
\usepackage{pstricks}
\usepackage{amsmath}
\usepackage{graphics}
\usepackage{mathrsfs}

\usepackage{fancyhdr}

\newtheorem{proof}{Proof}[subsubsection]
\newtheorem{theorem}{Theorem}
\newtheorem{definition}[theorem]{Definition}

\newtheorem{lemma}[theorem]{Lemma}
\newtheorem{remark}[theorem]{Remark}

\newcommand{\om}{\Omega}

\newcommand{\dt}{\mbox{ }dt}

\newcommand{\dc}{\|}

\newcommand{\bfz}{\mbox{\boldmath{$z$}}}

\newcommand{\bfU}{\mbox{\boldmath{$U$}}}

\newcounter{constants}
\begin{document}

\title{Global weak solutions for coupled transport processes
in concrete walls at high temperatures}

\author{
{Michal Bene\v{s}$^1$ and Radek \v{S}tefan$^2$}
\smallskip
\\
{\small $^1$Department of Mathematics}
\\
{\small $^2$Department of Concrete and Masonry Structures}
\smallskip
\\
{\small Faculty of Civil Engineering}
\\
{\small Czech Technical University in Prague}
\\
{\small Th\'{a}kurova 7, 166 29 Prague 6, Czech Republic}
\\
{\small e-mail: benes@mat.fsv.cvut.cz}}

\date{ }

\maketitle

\begin{abstract}

We consider an initial-boundary value problem for a fully nonlinear
coupled parabolic system with nonlinear boundary conditions
modelling hygro-thermal behavior of concrete at high temperatures.
We prove a global existence of a weak solution to this system on an
arbitrary time interval. The main result is proved by an
approximation procedure. This consists in proving the existence of
solutions to mollified problems using the Leray-Schauder theorem,
for which a priori estimates are obtained. The limit then provides a
weak solution for the original problem. A practical example
illustrates a performance of the model for a problem of a concrete
segment exposed to transient heating according to three different
fire scenarios. Here, the focus is on the short-term pore pressure
build up, which can lead to explosive spalling of concrete and
catastrophic failures of concrete structures.

\end{abstract}

\section{Introduction}\label{Intro}

 The hygro-thermal behavior
of concrete plays a crucial role in the assessment of the
reliability and lifetime of concrete structures. The impact of heat
and mass transfer processes become particularly evident at high
temperatures, where the increased pressure in pores, large
temperature gradients and temperature induced-creep may lead to
catastrophic service failures. Examples of such severe situations
include hypothetical nuclear reactor accidents \cite{Shekarchi2003},
fires in tunnels and tall buildings \cite{Sch2002} or simulation of
airfield concrete pavements \cite{Ju1998}. Since high-temperature
experiments are very expensive, predictive modeling of humidity
migration, pore pressure development and damage distribution can
result in significant economic savings. The first mathematical
models of concrete exposed to temperatures exceeding $100^{\circ}$C
were formulated by Ba\v{z}ant and~Thonguthai in \cite{BaTh1978}.
Since then, a considerable effort has been invested into detailed
numerical simulations of concrete structures subject to high
temperatures.
However, much less attention has been given to the
qualitative properties of the model.

Let us briefly present the single-fluid-phase, purely
macroscopic, model for prediction of hygro-thermal
behavior of heated concrete. The one-dimensional heat and mass transport in concrete wall
is governed by the following parabolic system \cite{Ba1997,BaTh1978,BaKa1996}:
\smallskip
\\
\emph{energy conservation equation for concrete:}
\begin{equation}\label{energy temp}
(C_w w \theta + \rho_S C_S \theta)_t =
(\Lambda(\theta,P)\theta_x)_x
 + (C_w \kappa(\theta,P)\theta P_x)_x - h_d d_t;
\end{equation}
\emph{water content conservation equation:}
\begin{equation}\label{conservation_mass}
w_t  = ( \kappa(\theta,P)P_x)_x + d_t.
\end{equation}

The primary unknowns in the model are temperature $\theta$, water
content $w$ (the mass of all free evaporable water per m$^3$ of
concrete) and pore pressure $P$. In order to complete the system
\eqref{energy temp}--\eqref{conservation_mass}, let us note that the
water content $w$ is connected with temperature $T$ and pore
pressure $P$ via sorption isotherms $\Phi=P/P_{sat}=\Phi(\theta,w)$
($\Phi$ denotes the relative humidity and $P_{sat}$ is saturation
vapor pressure), which need to be determined experimentally. The
thermal conductivity $\Lambda$ and permeability $\kappa$ are assumed
to be positive smooth functions of their arguments. Further, $C_w$
represents the isobaric heat capacity of moisture, $\rho_S$ and
$C_S$, respectively, are the mass density and the isobaric heat
capacity of solid microstructure (excluding hydrate water), $d$ is
the total mass of the free water released in the pores by
dehydration and $h_{d}$ denotes the enthalpy of dehydration per unit
mass.

To complete the model,
one should prescribe the appropriate boundary conditions across the
boundary and initial conditions on primary unknowns.
In the case of heat transfer through the boundary, the Neumann type
or radiative (for heat flux) boundary conditions,
respectively, are frequently under consideration in practice:
\begin{eqnarray}
-\Lambda(\theta,P)\theta_x \cdot  n &=&
\alpha_c(\theta-\theta_{\infty}), \label{neumann a}
\\
-\Lambda(\theta,P)\theta_x \cdot  n &=&
\alpha_c(\theta-\theta_{\infty})+e\sigma
(\theta^4 - \theta^4_{\infty}), \label{radiative a}
\end{eqnarray}
respectively, in which $\alpha_c$ designates the film coefficient for
the heat transfer, and $\theta_\infty$ is temperature of the environment.
The last expression in \eqref{radiative a} expresses the radiative
contribution to the heat flux, quantified by the Stefan-Boltzmann law in terms
of the relative surface emissivity $e$ and the Stefan-Boltzmann constant
$\sigma$ and the temperature difference $( \theta^4 -
\theta_\infty^4)$.

 The humidity flux across the boundary is
quantified by the Newton law:
\begin{equation}\label{neumann p}
-\kappa(\theta,P)P_x \cdot n = \beta_c (P - P_{\infty}),
\end{equation}
where the right hand side represents the humidity dissipated into the
surrounding medium with the water vapor pressure in the environment $P_\infty$
and $\beta_c$ represents the surface emissivity of water.

\medskip

The system
\eqref{energy temp}--\eqref{conservation_mass}, accompanied by the
appropriate boundary and initial conditions, represents a
challenging mathematical problem. Models of hygro-thermal transport
processes in concrete are associated to systems of doubly and
strongly nonlinear parabolic equations of type (written in operator
form)
$$
\partial_t \mathcal{B}(\bfU)-\nabla\cdot
\mathcal{A}(\bfU,\nabla \bfU)=\mathcal{F}(\bfU,\nabla \bfU).
$$
Although there is no general existence and regularity theory for
such problems, some partial outcomes assuming special structure of
operators $\mathcal{A}$ and $\mathcal{B}$ and growth
conditions on $\mathcal{F}$ can be found in
the literature. The existence of weak solutions subject to mixed
boundary conditions with homogeneous Neumann boundary conditions has
been shown by Alt \& Luckhaus in~\cite{AltLuckhaus1983}. They
obtained an existence result assuming the operator $\mathcal{B}$ in the
parabolic part to be only (weak) monotone and subgradient. This
result has been extended e.g. by Filo \& Ka\v{c}ur in
\cite{FiloKacur1995}, who proved the local existence of the weak
solution for the system with nonlinear Neumann boundary conditions
and under more general growth conditions on nonlinearities in $\bfU$.
These results, however, are not applicable if $\mathcal{B}$ does not take the
subgradient form, which is typical of coupled heat and mass
transport models. Thus, the analysis needs to exploit the specific
structure of the problem. In this context, Vala in \cite{Vala2002}
proved the existence of solution to the purely diffusive
hygro-thermal model with non-symmetry in the parabolic part but with
unrealistic symmetry in the elliptic term. In \cite{LiSun2010} and
\cite{LiSunWang2010}, Li~\emph{et al.} study a coupled model for
heat and mass transport arising from textile industry, which is
described by a degenerate and strongly coupled parabolic system.
They prove the global existence for one-dimensional problem using
the Leray-Schauder fixed point theorem.

Additional results are available for the time discrete setting.
Dal\'{i}k \emph{et al.}~\cite{DaDaVa2000} analyzed the numerical
solution of the Kiessl model for moisture and heat transfer in
porous materials. They proved some existence and regularity results
and suggested an efficient numerical approach to the solution of the
resulting system of highly non-liner equations. However, the Kiessl
model is valid for limited temperature range only and as such it is
inappropriate for high-temperature applications. In
\cite{Benes2011}, Bene\v{s} \emph{et al.} extended the
work~\cite{DaDaVa2000} by proving the existence of an approximate
solution for the Ba\v{z}ant--Thonguthai model, arising from the
semi-implicit discretization in time.

\smallskip

In the present paper we
prove a global-in-time existence of a weak solution to a fully nonlinear
coupled parabolic system with nonlinear boundary conditions
modeling heat and moisture transport in concrete walls at high temperatures based on
the single-fluid-phase model introduced by Ba\v{z}ant \& Thonguthai
\cite{BaTh1978}.  The main result is proved by means of a fixed point argument based
on the Leray-Schauder approach and approximation procedure and then carrying
out the passage to the limit.

\paragraph{Outline of the paper.}
The paper is organized in the following manner. In Section \ref{prel},
we introduce basic notation and the appropriate function spaces,
present the classical formulation of the problem under consideration
and specify our assumptions on data
and structure conditions in the model.
In Section~\ref{sec:main_result}, we formulate the problem in the
variational sense and state the main result of the paper -- the global-in-time
existence of the weak solution.
In Section \ref{sec:proof_main}, we derive the a priori estimates
for an approximate solution of the auxiliary regularized problem.
The solution of the regularized problem is obtained by the Leray-Schauder fixed point theorem.
Using the limiting procedure we prove the existence of the weak solution to the original problem.
In Section \ref{sec:numerical_exp}, numerical
experiments are performed to present the moisture migration,
temperature distribution and pore pressure build up in the model of
concrete wall one-side-exposed to various fire scenarios.

\section{Preliminaries}\label{prel}

\subsection{Notation and some function spaces}
\label{notation}

Vectors, vector functions and operators acting on vector functions
are denoted by~boldface letters. Throughout the paper, we will
always use positive constants $c$, $c_1$, $c_2$, $\dots$, which are
not specified and which may differ from line to line.
As usual, for a function $\phi=\phi(x,t)$,
$\phi_x$ and $\phi_t$ indicate the partial derivatives
with respect to spatial variable $x$ and temporal variable $t$.
 Let $T>0$ and $\ell >0$ be the fixed values, $\Omega = (0,\ell)$,
 $I=(0,T)$,
 $Q_{T}=\Omega\times I$, $Q_{t}=\Omega\times(0,t]$ ($t>0$).
 We denote by
$\mathbf{W}^{l,p}\equiv W^{l,p}(\Omega)^2$, $l\ge 0$ and $1 \leq p
\leq \infty$ and, especially, $\mathbf{L}^p\equiv\mathbf{W}^{0,p}$,
where $W^{k,p}(\Omega)$, $k\geq 0$, $p\in [1,+\infty]$, denotes the
usual Sobolev space with the norm $\|\cdot\|_{W^{k,p}(\Omega)}$ and,
in addition, $L^p(\Omega)$ denotes the usual Lebesgue space equipped
with the norm $\|\cdot\|_{L^p(\Omega)}$.

Let $X$ be an arbitrary Banach space with the norm $\|\cdot\|_X$
($X^*$ represents the dual space to $X$). Let $r\in[1,\infty)$. As
usual $L^r(I;X)$ and~ $L^{\infty}(I;X)$ denote the Banach spaces
$$
\Bigl\{\phi;\;\phi(t)\in X \mbox{ for almost every } t\in I,\,
\int_{0}^{T}\dc\phi(t)\dc_X^r\dt<\infty\Bigr\}
$$
and
$$
\Bigl\{\phi;\;\;\phi(t)\in X \mbox{ for almost every }
t\in I,\,\mbox{\rm ess}
\sup_{t\in I}\dc\phi(t)\dc_X <\infty\Bigr\}
$$
with the norms
$$
\dc\phi\dc_{L^p(I;X)} :=
\Bigl(\int_{0}^{T}\dc\phi(t)\dc_X^r\dt\Bigr)^{1/r}
$$
and
$$
\dc\phi\dc_{L^{\infty}(I;X)} := \mbox{\rm ess}
\sup_{t\in I}\dc\phi(t)\dc_X.
$$
Further we introduce the space $C(I;X)$ -- space of functions $\phi:[0,T]\rightarrow X$, continuous, for which
$$
\dc\phi\dc_{C(I;X)} := \mbox{\rm ess}
\sup_{t\in [0,T]}\dc\phi(t)\dc_X.
$$
Define the spaces
$$
{V}^{2,1}_2(Q_T)\equiv \left\{ \phi \in L^2(Q_T) |\, \phi_t,
\phi_x, \phi_{xx} \in L^2(Q_T)  \right\} \textmd{ and }\mathbf{V}^{2,1}_2(Q_T)\equiv {V}^{2,1}_2(Q_T)^2.
$$
We will often use the following well-known embeddings that are consequences of Aubin-Lions lemma
and interpolation inequalities:
\begin{eqnarray}
&&{V}^{2,1}_2(Q_T) \hookrightarrow\hookrightarrow
L^2(I;\mathbf{W}^{1,2}), \label{comb_emb01}
\\
&& L^2(I;W^{1,2}(\Omega)) \cap
L^{\infty}(I;L^2(\Omega)) \hookrightarrow L^6(Q_T),
\\
&&{V}^{2,1}_2(Q_T) \hookrightarrow C(\overline{Q}_T).
\end{eqnarray}

\subsection{Classical formulation of the problem}
\label{}

Incorporating the
relation  $P=P_{sat}(\theta)\Phi(\theta,w)\equiv \mathcal{P}(\theta,w)$ (via
sorption isotherms $\Phi=P/P_{sat}=\Phi(\theta,w)$) into the system \eqref{energy temp}--\eqref{neumann p}
we can eliminate the unknown field $P$ and consider the problem with
only two unknowns $\theta$ and $w$. Consequently,
the classical formulation of the problem we are going to study reads as follows:
\begin{align}
w_t - d_t & = \left( \delta_w(\theta,w)w_x\right)_x + \left(
\delta_{\theta}(\theta,w)\theta_x\right)_x  && {\rm in}\;
Q_{T},\label{eq1b}
\\
(C_w w \theta + \rho_S C_S \theta)_t + h_d d_t &=
(\lambda(\theta,w)\theta_x)_x
+ C_w (\theta  \left( \delta_w(\theta,w)w_x +
\delta_{\theta}(\theta,w)\theta_x\right)  )_x  &&{\rm in} \; Q_{T},
\\
d_t &= -\frac{1}{\tau}(d-d_{eq}(\theta))  &&{\rm in} \; Q_{T}, \label{eq1b_dehydr}
\\
[\delta_w(\theta,w)w_x + \delta_{\theta}(\theta,w)\theta_x ]\Big|_{x=0} &=
0 && {\rm in} \; I,
\\
[\lambda(\theta,w)\theta_x  ]\Big|_{x=0}  &= 0 && {\rm in} \;  I,
\\
[-\delta_w(\theta,w)w_x - \delta_{\theta}(\theta,w)\theta_x
 ]\Big|_{x=\ell}
&= \beta_c \left(\mathcal{P}(\theta(\ell,t),w(\ell,t)) -  P_{\infty}\right) && {\rm in} \; I,
\label{bc_w}
\\
[-\lambda(\theta,w)\theta_x  ]\Big|_{x=\ell}  &=
(\alpha_c+e\sigma|\theta(\ell,t)|^3)\theta(\ell,t)
 - \vartheta(t) \!\! && {\rm in} \;  I, \label{bc_radiative}
\\
d(x,0)&=0 &&{\rm in} \;  \Omega,
\\
\theta(x,0)&= \theta_0(x) &&{\rm in} \;  \Omega,
\\
w(x,0)&= w_0(x) &&{\rm in} \;  \Omega.\label{eq10b}
\end{align}
Here we assume that all functions are smooth enough.
The unknowns in the proposed model are temperature $\theta$, water
content $w$ and the function $d$ (the total mass of the free water released
in the pores by dehydration).
$\lambda$, $\delta_{\theta}$ and $\delta_{w}$ are diffusion
coefficient functions depending non-linearly on $\theta$ and $w$. In
\eqref{eq1b_dehydr} $\tau$ represents the characteristic time of
dehydration and the function $d_{eq}=d_{eq}(\theta)$ defines the
water mass created by dehydration at equilibrium at the given
temperature $\theta$ \cite{PontEhrlacher}.
In \eqref{bc_radiative} the function $\vartheta$ expresses thermal loading given by (cf. \eqref{radiative a})
\begin{equation}
\vartheta(t):=\alpha_c \theta_{\infty}(t)+ e\sigma\theta^4_{\infty}(t).
\end{equation}
Finally, $\theta_0$ and $w_0$ represent the initial distributions of the primary
unknowns $\theta$ and $w$, respectively.

\subsection{Structural conditions and assumptions on physical parameters}
\label{assumptions}

Here we specify assumptions on material coefficients and data in the model \eqref{eq1b}--\eqref{eq10b}.

\begin{itemize}
\item[$A_1$]
We assume that the material parameters $\rho_S$, $C_S$, $C_w$, $h_d$, $\tau$,
$\alpha_c$, $\beta_c$, $P_{\infty}$,
$\sigma$ and $e$ are real positive constants.
\item[$A_2$]   $\delta_{\theta}$, $\delta_{w}$ and $\lambda$ are $C^1$ functions,
$\delta_{\theta}$, $\delta_{w}$, $\lambda:\mathbb{R}^2\rightarrow \mathbb{R}$,
satisfying for certain
positive constants $\delta_0$, $\delta_1$, $\delta_2$, $\lambda_0$, $\lambda_1$ and $\lambda_2$
and for all $\bfz=[z_1,z_2] \in \mathbb{R}^2$
\begin{eqnarray}
0 < \delta_0 \leq \delta_w(\bfz) \leq \delta_1 < +\infty, &&
\label{con11a}
\\
|\partial_{z_1}\delta_w(\bfz)|,|\partial_{z_2}\delta_w(\bfz)| \leq \delta_2 < +\infty,  &&
\label{con11b}
\\
0 < \delta_0 \leq \delta_{\theta}(\bfz) \leq \delta_1 < +\infty, &&
\label{con11bb}
\\
|\delta_{\theta}(\bfz)| \leq c|z_2|, && \label{con11c}
\\
|\partial_{z_1}\delta_{\theta}(\bfz)|,|\partial_{z_2}\delta_{\theta}(\bfz)| \leq \delta_2 < +\infty, &&
\label{con11d}
\\
0<\lambda_0\leq \lambda(\bfz) \leq \lambda_1 < +\infty, &&
\label{con11e}
\\
 |\partial_{z_1}\lambda(\bfz)|,|\partial_{z_2}\lambda(\bfz)| \leq \lambda_2 < +\infty. &&
\label{con11f}
\end{eqnarray}

\item[$A_3$]
$d_{eq}:\mathbb{R}\rightarrow \mathbb{R}$ is a function of class $C^1$ and there exist positive constants $d_1$ and $d_2$ such that for every $z \in \mathbb{R}$
\begin{eqnarray}
0 \leq d_{eq}(z) \leq d_1 < +\infty, &&
\label{con12a}
\\
|d'_{eq}(z)| \leq d_2 < +\infty.
\label{con12b}
\end{eqnarray}

\item[$A_4$]
We assume $\mathcal{P}:\mathbb{R}^2\rightarrow \mathbb{R}$ to be a Lipschitz
continuous function in $\mathbb{R}^2$ such that $\mathcal{P}(\xi_1,\xi_2)
\cdot \xi_2 \geq 0$ for every $[\xi_1,\xi_2] \in \mathbb{R}^2$.

\item[$A_5$]
$\vartheta$ is a positive continuous function in $[0,T]$.

\item[$A_6$]
$\theta_0$ and $w_0$ are positive continuous functions in $\overline{\Omega}$.

\end{itemize}

\section{Main result}\label{sec:main_result}

The aim of this paper is to prove the existence of a weak solution to the problem described
by the system \eqref{eq1b}--\eqref{eq10b}.
Now we formulate our problem in the variational sense.
\begin{definition}\label{def_weak_solution}
The triplet $[\theta,w,d]\in L^2(I;\mathbf{W}^{1,2})\times
C(I;L^{\infty}(\Omega))$ is called the weak solution to the system
\eqref{eq1b}--\eqref{eq10b} iff
\begin{align}\label{weak_form}
&-\int_{Q_T} \left( w - d \right) \phi_t +\left( C_w \theta w +
\rho_S C_S \theta + h_d d \right) \psi_t {\rm d}x{\rm d}t
\nonumber\\
& \qquad +\int_{Q_T} \left( \delta_w(\theta,w) w_x +
\delta_{\theta}(\theta,w) \theta_x \right) \phi_x  {\rm d}x{\rm d}t
+\int_0^T \beta_c \left(\mathcal{P}(\theta(\ell,t),w(\ell,t)) -  P_{\infty}\right) \phi(\ell,t) \;{\rm d}t
\nonumber
\\
& \qquad +\int_{Q_T}  \lambda(\theta,w) \theta_x \psi_x \;{\rm d}x{\rm d}t
%
%
-\int_{Q_T} \theta\left( \delta_w(\theta,w) w_x +
\delta_{\theta}(\theta,w)\theta_x \right) \psi_x \;{\rm d}x{\rm d}t
\nonumber\\
& \qquad +\int_0^T \!\! \left[(\alpha_c+e\sigma|\theta(\ell,t)|^3)\theta(\ell,t)
 - \vartheta(t) +  C_w \theta(\ell,t) \beta_c \left(\mathcal{P}(\theta(\ell,t),w(\ell,t)) -  P_{\infty}\right)\right] \psi(\ell,t)   {\rm d}t
\nonumber\\
&  = \int_{\Omega} (w_0-d_0)\phi(x,0) + \left( C_w \theta_0 w_0 + \rho_S
C_S \theta_0 + h_d d_0 \right)\psi(x,0) {\rm d}x
\end{align}
holds for all test functions $[\phi,\psi]\in
C^{\infty}(\overline{Q}_T)$, $\phi(x,T)=\psi(x,T)=0$ $\forall x \in
\Omega$. Here $d$ satisfies the equation
\begin{equation}\label{weak_solution_d}
d(x,t) = \frac{1}{\tau}\int_0^t e^{(s-t)/{\tau}} d_{eq}(\theta(x,s)){\rm d}s
\end{equation}
for every $t \in (0,T)$ and almost every $x \in \Omega$.
\end{definition}

\begin{remark}
{The coupled problem described by the system \eqref{eq1b}--\eqref{eq10b},
being intensively used in practise, does not seem analytically studied so far.
To the best of our knowledge, there is no existence result for
the presented model available.}
\end{remark}

The main result of this paper reads:
\begin{theorem}[Main result]\label{main_result}
Let the assumptions $A_1$--$A_6$ be satisfied. Then there exists at
least one weak solution $[\theta,w,d]\in
L^2(I;\mathbf{W}^{1,2})\times C(I;L^{\infty}(\Omega))$ of the
system \eqref{eq1b}--\eqref{eq10b} in the sense of Definition
\ref{def_weak_solution}.
\end{theorem}

\section{Proof of the main result}\label{sec:proof_main}

\subsection{The regularized problem}

In the proof of the main result we use the method of mollification in order to get
the approximate solution of the problem \eqref{eq1b}--\eqref{eq10b}.
Let us introduce the mollified system
\begin{align}
w_t - \left(
\delta_w(\theta^{\varepsilon},w^{\varepsilon})w_x\right)_x - \left(
\delta_{\theta}(\theta^{\varepsilon},w)\theta^{\varepsilon}_x\right)_x
 & = d_t
&& {\rm in}\; Q_{T},
\label{eq1d}
\\
( C_w w \theta +  \rho_S C_S \theta)_t + h_d d_t &=
(\lambda(\theta^{\varepsilon},w^{\varepsilon})\theta_x)_x
\nonumber
\\
& \quad +
C_w\left(\theta(\delta_w(\theta^{\varepsilon},w^{\varepsilon}) w_x +
\delta_{\theta}(\theta^{\varepsilon},w)\theta^{\varepsilon}_x )
\right)_x &&{\rm in} \; Q_{T}, \label{eq2d}
\\
d_t &= - \frac{1}{\tau}(d-d_{eq}(\theta))  &&{\rm in} \; Q_{T}.
\label{eq3d}
\end{align}
Here $\theta^{\varepsilon}:=J_{\varepsilon}\star \mathcal{E}(\theta)$
is a regularization of $\theta$, $\varepsilon$ is a small positive real number, $J_{\varepsilon}$ denotes the standard mollifier defined for the function $\theta$ (see \cite[Paragraph 2.28 and Theorem 2.29]{AdamsFournier1992}) and $\mathcal{E}(\theta)$ is the extension operator extending $\theta$ to be zero on $\mathbb{R}^2\setminus Q_T$ (see \cite[Paragraph 5.17]{AdamsFournier1992}).
Further we set $w^{\varepsilon}:=J_{\varepsilon}\star w$ and in a similar way we introduce $\theta_0^{\varepsilon}$, $w_0^{\varepsilon}$ and $\vartheta^{\varepsilon}$, respectively, as regularizations of $\theta_0$, $w_0$ and $\vartheta$, respectively.

Substituting \eqref{eq1d} into \eqref{eq2d}, equations
\eqref{eq1d}--\eqref{eq3d} can be rewritten as
\begin{align}
w_t - \left(
\delta_w(\theta^{\varepsilon},w^{\varepsilon})w_x\right)_x - \left(
\delta_{\theta}(\theta^{\varepsilon},w)\theta^{\varepsilon}_x\right)_x
 & =  d_t
&& {\rm in}\; Q_{T},\label{eq30}
\\
(C_w w + \rho_S C_S) \theta_t  + (C_w \theta + h_d)d_t &=
(\lambda(\theta^{\varepsilon},w^{\varepsilon})\theta_x)_x \nonumber
\\
& \quad + C_w
\left(\delta_w(\theta^{\varepsilon},w^{\varepsilon})w_x +
\delta_{\theta}(\theta^{\varepsilon},w)\theta^{\varepsilon}_x
\right)\theta_x &&{\rm in} \; Q_{T},\label{eq31}
\\
d_t &= -\frac{1}{\tau}(d-d_{eq}(\theta))  &&{\rm in} \; Q_{T}.
\end{align}
The governing equations are supplemented by the following regularized boundary and initial conditions
\begin{align}
[\delta_w(\theta^{\varepsilon},w^{\varepsilon})w_x +
\delta_{\theta}(\theta^{\varepsilon},w)\theta^{\varepsilon}_x
]\Big|_{x=0} &= 0 && {\rm in} \; I,\label{eq39c}
\\
[\lambda(\theta^{\varepsilon},w^{\varepsilon})\theta_x ]\Big|_{x=0}  &= 0
&& {\rm in} \;  I,
\\
[-\delta_w(\theta^{\varepsilon},w^{\varepsilon})w_x -
\delta_{\theta}(\theta^{\varepsilon},w)\theta^{\varepsilon}_x
]\Big|_{x=\ell} &=
\beta_c \left(\mathcal{P}(\theta^{\varepsilon}(\ell,t),w(\ell,t)) -  P_{\infty}\right) && {\rm in}
\; I,
\\
[-\lambda(\theta^{\varepsilon},w^{\varepsilon})\theta_x ]\Big|_{x=\ell} &=
(\alpha_c+e\sigma|\theta^{\varepsilon}(\ell,t)|^3)\theta(\ell,t)
 - \vartheta^{\varepsilon}(t) \!\!\!&& {\rm in} \;  I,
\\
d(x,0)&=0 &&{\rm in} \;  \Omega,
\\
\theta(x,0)&=\theta_0^{\varepsilon}(x) &&{\rm in} \;  \Omega,
\\
w(x,0)&=w_0^{\varepsilon}(x)  &&{\rm in} \; \Omega. \label{eq39}
\end{align}

In this subsection we prove the existence of a strong solution to the regularized problem
\eqref{eq30}--\eqref{eq39}.

\begin{theorem}\label{th_sup_5}
The system \eqref{eq30}--\eqref{eq39} has a strong solution
$[\theta_{\varepsilon},w_{\varepsilon}] \in
\mathbf{V}^{2,1}_2(Q_T)$, such that
\begin{eqnarray}
\|[\theta_{\varepsilon},w_{\varepsilon}]\|_{L^2(I;\mathbf{W}^{1,2})}
&\leq& c ,\label{est_reg_solution_1}
\\
\|[\theta_{\varepsilon},w_{\varepsilon}]\|_{L^{\infty}(I;\mathbf{L}^2)}
&\leq& c ,\label{est_reg_solution_2}
\end{eqnarray}
where the constant $c$ does not depend on $\varepsilon$.
\end{theorem}

\begin{proof} Proof of Theorem \ref{th_sup_5} follows from the a priori
estimates in the ``weak'' classes of functions by the usual
Leray-Schauder fixed point arguments.

\paragraph{Solution to an auxiliary problem.}
For any given couple $[\tilde{\theta},\tilde{w}]\in
L^2(I;\mathbf{W}^{1,2})$ and $0 \leq \zeta \leq 1$ consider the
initial-boundary value problem

\begin{align}
w_t - \left(
\delta_w(\tilde{\theta}^{\varepsilon},\tilde{w}^{\varepsilon})w_x\right)_x
- \left(
\delta_{\theta}(\tilde{\theta}^{\varepsilon},w)\tilde{\theta}^{\varepsilon}_x\right)_x
 & = d_t
&& {\rm in}\; Q_{T},
\label{eq40}
\\
(C_w w + \rho_S C_S) \theta_t  + (C_w \theta + h_d)d_t &=
(\lambda(\tilde{\theta}^{\varepsilon},\tilde{w}^{\varepsilon})\theta_x)_x
\nonumber
\\
& \quad + C_w
\left(\delta_w(\tilde{\theta}^{\varepsilon},\tilde{w}^{\varepsilon})w_x
+
\delta_{\theta}(\tilde{\theta}^{\varepsilon},w)\tilde{\theta}^{\varepsilon}_x
\right)\theta_x &&{\rm in} \; Q_{T},
\label{eq41}
\\
d_t &= -\frac{1}{\tau}(d- \zeta d_{eq}(\tilde{\theta}))  &&{\rm in}
\; Q_{T},
\label{eq42}
\\
[\delta_w(\tilde{\theta}^{\varepsilon},\tilde{w}^{\varepsilon})w_x +
\delta_{\theta}(\tilde{\theta}^{\varepsilon},w)\tilde{\theta}^{\varepsilon}_x]
\Big|_{x=0} &= 0 && {\rm in} \; I,
\label{eq43}
\\
[\lambda(\tilde{\theta}^{\varepsilon},\tilde{w}^{\varepsilon})\theta_x]
\Big|_{x=0}  &= 0 && {\rm in} \;  I,
\label{eq44}
\\
[-\delta_w(\tilde{\theta}^{\varepsilon},\tilde{w}^{\varepsilon})w_x -
\delta_{\theta}(\tilde{\theta}^{\varepsilon},w)\tilde{\theta}^{\varepsilon}_x]
\Big|_{x=\ell} &
=
\beta_c(\mathcal{P}(\tilde{\theta}^{\varepsilon}(\ell,t),w(\ell,t)) - \zeta P_{\infty})
&& {\rm in} \; I,
\label{eq45}
\\
[-\lambda(\tilde{\theta}^{\varepsilon},\tilde{w}^{\varepsilon})\theta_x]
\Big|_{x=\ell}
&  = \! (\alpha_c \!+\! e\sigma|\tilde{\theta}^{\varepsilon}(\ell,t)|^3)\theta(\ell,t)
 \!-\! \zeta\vartheta^{\varepsilon}(t)\!\!\!\! && {\rm in} \;
I,
\label{eq46}
\\
d(x,0)&= 0 &&{\rm in} \;  \Omega,
\label{eq47}
\\
\theta(x,0)&=\zeta\theta_0^{\varepsilon}(x) &&{\rm in} \;
\Omega,
\label{eq48}
\\
w(x,0)&=\zeta w_0^{\varepsilon}(x)  &&{\rm in} \; \Omega.
\label{eq49}
\end{align}
The proof of the existence of the solution to the system
\eqref{eq40}--\eqref{eq49} is split into three steps:
\paragraph{Step 1}
First we treat the initial problem \eqref{eq42} and \eqref{eq47}.
Let $d$ be a solution of the ordinary differential equation
\begin{displaymath}
d_t +\frac{1}{\tau}d  =  \frac{1}{\tau}\zeta d_{eq}(\tilde{\theta})
\end{displaymath}
(which holds for almost every $t \in (0,T)$ and $x \in \Omega$) with the initial
condition
\begin{displaymath}
d(x,0) = 0 \; \textmd{ in }\Omega.
\end{displaymath}
Then
\begin{equation}\label{solution_d}
d(x,t) = \frac{1}{\tau}\int_0^t e^{(s-t)/{\tau}}\zeta d_{eq}(\tilde{\theta}(x,s)){\rm d}s
\end{equation}
for every $t \in (0,T)$ and almost every $x \in \Omega$.
Hence, for almost every $x\in\Omega$ we have $d(x,\cdot)\in W^{1,\infty}(I)$.
Moreover, $d$ and $d_t$ are nonnegative almost everywhere in $Q_T$
($d(x,\cdot)$ is nondecreasing with respect to $t$).

\paragraph{Step 2}
The system of equations \eqref{eq40},
\eqref{eq43},  \eqref{eq45} and \eqref{eq49} (written in non-divergence form)
is a special case of the problem (cf. \cite[Chapter V., (7.1)--(7.2)]{LadSolUr})
\begin{align*}
w_t-a(x,t)w_{xx} - b(x,t,w,w_x) & = 0 &&
{\rm in } \; Q_T, \label{}
\\
 a(0,t)w_x(0,t) + g_1(0,t,w)   &= 0
&& {\rm in} \;  I,
\\
-a(\ell,t)w_x(\ell,t) + g_2(\ell,t,w) & = 0 && {\rm in} \;  I,
\\
w(x,0)&=\zeta w_0^{\varepsilon}(x)  &&{\rm in} \; \Omega. \label{}
\end{align*}
In our case we have
\begin{eqnarray*}
a(x,t) &=&    \delta_w(\tilde{\theta}^{\varepsilon},\tilde{w}^{\varepsilon}),
\\
b(x,t,w,w_x)
 & = &
\left(
\delta_w(\tilde{\theta}^{\varepsilon},\tilde{w}^{\varepsilon})\right)_x w_x
+ \left(
\delta_{\theta}(\tilde{\theta}^{\varepsilon},w)\right)_x\tilde{\theta}^{\varepsilon}_x
+
\delta_{\theta}(\tilde{\theta}^{\varepsilon},w)\tilde{\theta}^{\varepsilon}_{xx}
+ d_t,
\\
 g_1(0,t,w)  & = &
 \delta_{\theta}(\tilde{\theta}^{\varepsilon}(0,t),w(0,t))\tilde{\theta}^{\varepsilon}_x (0,t),
\\
 g_2(\ell,t,w)  & = &
-
\delta_{\theta}(\tilde{\theta}^{\varepsilon}(\ell,t),w(\ell,t))\tilde{\theta}^{\varepsilon}_x(\ell,t)
-
\beta_c\mathcal{P}(\tilde{\theta}^{\varepsilon}(\ell,t),w(\ell,t))
+
\beta_c \zeta P_{\infty}.
\end{eqnarray*}
Under the assumptions $A_1$--$A_4$ and $A_6$,
following the classical parabolic--equation theory for quasi-linear equations
\cite{LadSolUr}, for any given $[\tilde{\theta},\tilde{w}]\in
L^2(I;\mathbf{W}^{1,2})$ and $d$ given by \eqref{solution_d} the problem \eqref{eq40},
\eqref{eq43},  \eqref{eq45} and  \eqref{eq49} admits the unique
 solution $w \in
{V}^{2,1}_2(Q_T)$ (see \cite[Chapter V., Theorem 7.4]{LadSolUr}).
Note that ${V}^{2,1}_2(Q_T) \hookrightarrow C(\overline{Q}_T)$. Now
we prove that $w$ is non-negative in $Q_T$. Let $w^+ =
\max\left\{0,w\right\}$ and $w^- = \max\left\{0,-w\right\}$. Test
\eqref{eq40} by $w^-$ to get
\begin{align}\label{eq50}
&\frac{1}{2}\int_{\Omega} |w^-(t)|^2 {\rm d}{x} + \int_{Q_t}
\delta_w(\tilde{\theta}^{\varepsilon},\tilde{w}^{\varepsilon})|w^-_x|^2
{\rm d}x{\rm d}s
+ \int_{Q_t} d_t w^- {\rm d}x{\rm d}s
\nonumber\\
& \qquad + \int_0^t \beta_c\zeta P_{\infty}w^-(\ell,s) {\rm d}s
- \int_0^t \beta_c\mathcal{P}(\tilde{\theta}^{\varepsilon}(\ell,s),w(\ell,s))w^-(\ell,s) {\rm d}s
\nonumber\\
& =  -\int_{Q_t}
\delta_{\theta}(\tilde{\theta}^{\varepsilon},w)\tilde{\theta}^{\varepsilon}_x
w^-_x {\rm d}x{\rm d}s.
\end{align}
Since $w=w^+ - w^-$  we can write
\begin{eqnarray*}
\int_0^t \mathcal{P}(\tilde{\theta}^{\varepsilon}(\ell,s),w(\ell,s))w^-(\ell,s) {\rm d}s
&=&
\int_0^t \mathcal{P}(\tilde{\theta}^{\varepsilon}(\ell,s),-w^-(\ell,s))w^-(\ell,s) {\rm d}s
\\
&=&
-\int_0^t \mathcal{P}(\tilde{\theta}^{\varepsilon}(\ell,s),-w^-(\ell,s))(-w^-(\ell,s)) {\rm d}s
\end{eqnarray*}
and due to $A_4$ the last integral on the left--hand side of the
equation \eqref{eq50} is always non-positive. Hence, applying
\eqref{con11c} and the Young's inequality to the term on the right
hand side and neglecting the non-negative terms on the left-hand
side, we arrive at the inequality
\begin{eqnarray*}
\frac{1}{2}\int_{\Omega} |w^-(t)|^2 {\rm d}{x}  + \int_{Q_t}
\delta_w(\tilde{\theta}^{\varepsilon},\tilde{w}^{\varepsilon})|w^-_x|^2
{\rm d}x{\rm d}s
&\leq&
\int_{Q_t}
c |w^-| |\tilde{\theta}^{\varepsilon}_x|
|w^-_x| {\rm d}x{\rm d}s
\\
&\leq&
c(\xi) \|\tilde{\theta}^{\varepsilon}_x\|_{L^{\infty}(Q_t)} \int_{Q_t}
|w^-|^2 {\rm d}x{\rm d}s
\\
&& \qquad + \xi
\|\tilde{\theta}^{\varepsilon}_x\|_{L^{\infty}(Q_t)} \int_{Q_t}
|w^-_x|^2 {\rm d}x{\rm d}s.
\end{eqnarray*}
Hence, choosing $\xi$ sufficiently small and taking into account \eqref{con11a}, we can use the Gronwall lemma to
conclude $w^- \equiv 0$ and thus $w \geq 0$ in $Q_T$.

\paragraph{Step 3}
Now having  $w \in {V}^{2,1}_2(Q_T)$, $w \geq 0$,
we get $\theta$ as the solution of the system
\eqref{eq41}, \eqref{eq44}, \eqref{eq46} and \eqref{eq48}.
We can write this linear parabolic problem in the form
\begin{align}
\theta_t-a_1(x,t)\theta_{xx} - a_2(x,t)\theta_x + a_3(x,t)\theta &= f(x,t)  &&\!
{\rm in } \; Q_T, \label{eqlin01}
\\
\theta_x(0,t)   &= 0
&& {\rm in} \;  I,
\\
-\lambda(\tilde{\theta}^{\varepsilon}(\ell,t),\tilde{w}^{\varepsilon}(\ell,t))\theta_x(\ell,t)
-(\alpha_c \!+\! e\sigma|\tilde{\theta}^{\varepsilon}(\ell,t)|^3)\theta(\ell,t)
& = \zeta\vartheta^{\varepsilon}(t)   && {\rm in} \;  I,
\\
\theta(x,0)&=\zeta\theta_0^{\varepsilon}(x)  &&{\rm in} \; \Omega,
\label{eqlin02}
\end{align}
where
\begin{eqnarray*}
a_1(x,t) &=& \frac{\lambda(\tilde{\theta}^{\varepsilon},\tilde{w}^{\varepsilon})}{(C_w w + \rho_S C_S)} ,
\\
a_2(x,t) &=& \frac{1}{(C_w w + \rho_S C_S)}
\left[ \lambda_x(\tilde{\theta}^{\varepsilon},\tilde{w}^{\varepsilon})
+
C_w
\left(\delta_w(\tilde{\theta}^{\varepsilon},\tilde{w}^{\varepsilon})w_x
+
\delta_{\theta}(\tilde{\theta}^{\varepsilon},w)\tilde{\theta}^{\varepsilon}_x
\right)
\right],
\\
a_3(x,t) &=&  \frac{C_w  d_t}{(C_w w + \rho_S C_S)},
\\
f(x,t) &=&  \frac{ - h_d d_t }{(C_w w + \rho_S C_S)}
\end{eqnarray*}
and $a_1 \in C(\overline{Q}_T)$, $a_2 \in L^{2+\epsilon}(I;L^{\infty}(\Omega))$ ($\epsilon$ is a small positive number), $a_3$ and $f \in L^{\infty}(\Omega)$.
By the linear theory for parabolic problems (see \cite[Theorem 2.1]{Denk}) there exists the uniquely determined
solution $\theta \in {V}^{2,1}_2(Q_T)$ of the problem
\eqref{eqlin01}--\eqref{eqlin02}.

\smallskip

Finally, let us conclude that for any given couple $[\tilde{\theta},\tilde{w}]\in
L^2(I;\mathbf{W}^{1,2})$ and $0 \leq \zeta \leq 1$ we have
$[\theta,w] \in \mathbf{V}^{2,1}_2(Q_T)$ as the solution of
the problem \eqref{eq40}--\eqref{eq49}.

\paragraph{Basic a priori estimates.}

In this paragraph we prove some a priori estimates for $\theta$ and $w$.

Test \eqref{eq30} by $C_w \theta^2$ and \eqref{eq31} by $2\theta$ to
obtain (adding these both resulting equations)
\begin{align}\label{eq65}
&\int_{Q_t} \left(\theta^2 (C_w w + \rho_S C_S )\right)_s {\rm
d}{x}{\rm d}s + \int_{Q_t}
2\lambda({\theta}^{\varepsilon},{w}^{\varepsilon})|\theta_x|^2 {\rm
d}x{\rm d}s + \int_{Q_t} C_w \theta^2 d_s {\rm d}x{\rm d}s \nonumber
\\
& \qquad + \int_{Q_t} 2 h_d \theta d_s {\rm d}x{\rm d}s
 + \int_0^t C_w \theta(\ell,s)^2
 \beta_c(\mathcal{P}(\theta^{\varepsilon}(\ell,s),w(\ell,s)) - \zeta P_{\infty})
  \;{\rm d}s
\nonumber\\
& \qquad + \int_0^t 2\theta(\ell,s) \left( (\alpha_c +
e\sigma|\theta^{\varepsilon}(\ell,s)|^3)\theta(\ell,s)
 - \zeta\vartheta^{\varepsilon}(s) \right) {\rm d}s
 =  0.
\end{align}
Hence, simple modifications yield
\begin{align}\label{eq65b}
&\int_{\Omega} \theta(x,t)^2 (C_w w(x,t) + \rho_S C_S ) {\rm d}{x} +
\int_{Q_t}
2\lambda({\theta}^{\varepsilon},{w}^{\varepsilon})|\theta_x|^2 {\rm
d}x{\rm d}s + \int_{Q_t} C_w \theta^2 d_s {\rm d}x{\rm d}s
\nonumber\\
& \qquad + \int_0^t C_w \theta(\ell,s)^2
 \beta_c\mathcal{P}(\theta^{\varepsilon}(\ell,s),w(\ell,s)){\rm d}s
+ \int_0^t 2\theta^2(\ell,s) (\alpha_c +
e\sigma|\theta^{\varepsilon}(\ell,s)|^3)
 {\rm d}s
\nonumber\\
& =  \int_{\Omega} \theta(x,0)^2 (C_w w(x,0) + \rho_S C_S ) {\rm d}{x}
 - \int_{Q_t} 2 h_d \theta d_s {\rm d}x{\rm d}s
\nonumber  \\
& \qquad + \zeta \int_0^t C_w \theta(\ell,s)^2 \beta_c
P_{\infty}{\rm d}s + \zeta \int_0^t 2\theta(\ell,s)
\vartheta^{\varepsilon}(s){\rm d}s.
 \end{align}
The last three integrals can be estimated using Young's inequality
and \cite[Remark 4]{FiloKacur1995}, respectively, as follows:
\begin{equation}\label{sup_01}
\int_{Q_t} 2 h_d \theta d_s {\rm d}x{\rm d}s \leq  h_d\int_{Q_t}
\theta^2{\rm d}x{\rm d}s + h_d \int_{Q_t} d_s^2(x,s) {\rm d}x{\rm d}s,
\end{equation}
\begin{equation}\label{sup_02}
\int_0^t C_w \theta(\ell,s)^2 \beta_c  P_{\infty}{\rm d}s \leq C_w
\beta_c  P_{\infty}  \int_{Q_t}  \epsilon|\theta_x|^2 +
C(\epsilon)|\theta|^2 {\rm d}x{\rm d}s
\end{equation}
and
\begin{equation}\label{sup_03}
\int_0^t 2\theta(\ell,s) \vartheta^{\varepsilon}(s){\rm d}s \leq
\epsilon \int_0^t \theta^2(\ell,s){\rm d}s + C(\epsilon)\int_0^t
\vartheta^{\varepsilon}(s)^2{\rm d}s,
\end{equation}
where $\epsilon$ represents sufficiently small positive real number.
Note that since $w \geq 0$, $A_4$ yields
\begin{displaymath}
\int_0^t C_w \theta(\ell,s)^2
 \beta_c\mathcal{P}(\theta^{\varepsilon}(\ell,s),w(\ell,s)){\rm d}s  \geq 0.
\end{displaymath}
Taking into account \eqref{con11e}, \eqref{sup_01}--\eqref{sup_03} and neglecting
the non-negative terms on the left-hand side in \eqref{eq65b} we
arrive at the estimate
\begin{equation}\label{eq65c}
c_1 \int_{\Omega} \theta(x,t)^2 {\rm d}{x}
+
c_2 \int_{Q_t}|\theta_x|^2 {\rm d}x{\rm d}s
+
c_3 \int_0^t \theta^2(\ell,s) {\rm d}s
\leq c_4 + c_5 \int_{Q_t} \theta^2{\rm d}x{\rm d}s.
\end{equation}
Hence, by means of Gronwall's lemma one checks that
\begin{eqnarray}
\| \theta \|_{L^2(I;W^{1,2}(\Omega))} &\leq& c,\label{est1_theta_a}
\\
\| \theta \|_{L^{\infty}(I;L^2(\Omega))} &\leq& c \label{est1_theta_b}
\end{eqnarray}
for the positive constant $c$ being independent of $\zeta$ and $\varepsilon$.

Now let us derive some uniform estimates for water content $w$.
Multiplying the equation \eqref{eq30} by $w$ and integrating over
$Q_t$ we get
\begin{align}\label{eq61}
&\frac{1}{2}\int_{\Omega}w(x,t)^2{\rm d}{x} + \int_{Q_t}
\delta_w({\theta}^{\varepsilon},{w}^{\varepsilon})|w_x|^2 {\rm
d}x{\rm d}s + \int_{Q_t}
\delta_{\theta}({\theta}^{\varepsilon},w){\theta}^{\varepsilon}_x
w_x {\rm d}x{\rm d}s
\nonumber \\
& \qquad
+  \int_0^t w(\ell,s)  \beta_c(\mathcal{P}(\theta^{\varepsilon}(\ell,s),w(\ell,s))-\zeta P_{\infty})\;{\rm d}s
\nonumber \\
& = \frac{1}{2}\int_{\Omega}w(x,0)^2{\rm d}{x} + \int_{Q_t} w\, d_s \,
{\rm d}x{\rm d}s
\end{align}
and consequently
\begin{align}\label{eq62}
& \frac{1}{2}\int_{\Omega}w(x,t)^2{\rm d}{x}
+\int_{Q_t}\delta_w({\theta}^{\varepsilon},{w}^{\varepsilon})|w_x|^2 {\rm
d}x{\rm d}s
+
\int_0^t w(\ell,s)  \beta_c\mathcal{P}(\theta^{\varepsilon}(\ell,s),w(\ell,s)){\rm d}s
\nonumber\\
& \leq
\frac{1}{2}\int_{\Omega}w(x,0)^2{\rm d}{x}
+
\frac{1}{2}\int_{Q_t} w^2 {\rm d}x{\rm d}s + \frac{1}{2}\int_{Q_t} d_s^2(x,s) {\rm d}x{\rm d}s
\nonumber\\
& \qquad -
\int_{Q_t}
\delta_{\theta}({\theta}^{\varepsilon},w){\theta}^{\varepsilon}_x
w_x {\rm d}x{\rm d}s
+
\zeta  \int_0^t w(\ell,s)  \beta_cP_{\infty}{\rm d}s.
\end{align}

By $A_4$ we deduce
\begin{equation}\label{sup_20a}
\int_0^t w(\ell,s)  \beta_c\mathcal{P}(\theta^{\varepsilon}(\ell,s),w(\ell,s)){\rm d}s \geq 0.
\end{equation}

Further, by \eqref{con11bb} and the Young's inequality we have
\begin{eqnarray}\label{sup_20b}
\int_{Q_t}
\delta_{\theta}({\theta}^{\varepsilon},w){\theta}^{\varepsilon}_x
w_x {\rm d}x{\rm d}s &\leq& \epsilon \int_{Q_t} w_x^2  {\rm d}x{\rm
d}s + C(\epsilon) \int_{Q_t}
|\delta_{\theta}({\theta}^{\varepsilon},w)|^2
|{\theta}^{\varepsilon}_x|^2 {\rm d}x{\rm d}s
\nonumber\\
&\leq&
 \epsilon \int_{Q_t} w_x^2  {\rm d}x{\rm d}s
+ c_1 C(\epsilon) \| \theta \|^2_{L^2(I;W^{1,2}(\Omega))},
\end{eqnarray}
while using the Cauchy's inequality and \cite[Remark
4]{FiloKacur1995} one obtains the estimate
\begin{eqnarray}\label{sup_20c}
 \int_0^t w(\ell,s)  \beta_cP_{\infty}{\rm d}s
&\leq&
c_1 + c_2 \int_0^t w^2(\ell,s) {\rm d}s
\nonumber \\
&\leq& c_1  +  c_2 \int_{Q_t}  \epsilon w_x^2 + C(\epsilon)w^2 {\rm
d}x{\rm d}s.
\end{eqnarray}
Now taking \eqref{eq62}--\eqref{sup_20c} together and using \eqref{con11a} and \eqref{est1_theta_a},
we arrive at the estimate of the form
\begin{equation}\label{eq66}
c_1 \int_{\Omega} w(x,t)^2 {\rm d}{x}
+
c_2 \int_{Q_t} w_x^2 {\rm d}x{\rm d}s
+
c_3 \int_{Q_t} w^2{\rm d}x{\rm d}s
\leq
c_4 + c_5 \int_{Q_t} w^2{\rm d}x{\rm d}s.
\end{equation}
By the Gronwall's inequality we get the uniform estimates for $w$
\begin{eqnarray}
\| w \|_{L^2(I;W^{1,2}(\Omega))} &\leq& c,\label{est1_w_a}
\\
\| w \|_{L^{\infty}(I;L^2(\Omega))} &\leq& c \label{est1_w_b}
\end{eqnarray}
for the positive constant $c$ being independent of $\zeta$ and $\varepsilon$.

Having established the a priori estimates for the solution
of the regularized problem we are ready to complete the proof
of Theorem \ref{th_sup_5} using the Leray--Schauder approach.


\paragraph{Leray-Schauder fixed point arguments.}

Denote by $X=L^2(I;\mathbf{W}^{1,2})$, take arbitrary
$[\tilde{\theta},\tilde{w}]\in X$, $\zeta\in[0,1]$ and define the
couple $[\theta,w] \in \mathbf{V}^{2,1}_2(Q_T)$ as the solution of
the problem \eqref{eq40}--\eqref{eq49}. Define the nonlinear mapping
$\mathcal{A}:X\times [0,1]\rightarrow X$, given by the equation
$[\theta,w]=\mathcal{A}([\tilde{\theta},\tilde{w},\zeta])$. It is a
technical procedure to check that the mapping $\mathcal{A}$ is
continuous and compact (cf. \eqref{comb_emb01}). For $\zeta=0$ we
have $\mathcal{A}([\tilde{\theta},\tilde{w},0])=[0,0]$ for all
$[\tilde{\theta},\tilde{w}]\in X$. The estimates
\eqref{est1_theta_a} and \eqref{est1_w_a} imply that $[\theta,w]$,
the solution of the problem $[\theta,w]
=\mathcal{A}([\theta,w,\zeta])$ for some $\zeta \in [0,1]$, is
uniformly bounded in $X$. Now the existence of at least one fixed
point $[\theta_{\varepsilon},w_{\varepsilon}] \in X$,
$\mathcal{A}[\theta_{\varepsilon},w_{\varepsilon},1]=
[\theta_{\varepsilon},w_{\varepsilon}]$, follows from the
Leray-Schauder theorem \cite{LeraySchauder1934}. Consequently,
$[\theta_{\varepsilon},w_{\varepsilon}]
=\mathcal{A}([\theta_{\varepsilon},w_{\varepsilon},1])\in
\mathbf{V}^{2,1}_2(Q_T)$ and
$[\theta_{\varepsilon},w_{\varepsilon}]$ is the strong solution of
the system \eqref{eq30}--\eqref{eq39}. The estimates
\eqref{est1_theta_a} and \eqref{est1_w_a} imply
\eqref{est_reg_solution_1}. Clearly, \eqref{est1_theta_b} and
\eqref{est1_w_b} yield \eqref{est_reg_solution_2}. The proof of
Theorem \ref{th_sup_5} is complete.
\end{proof}

\subsection{Passage to the limit for $\varepsilon \rightarrow 0$}
\label{}

To complete the proof of the main result stated in
Theorem \ref{main_result} we pass to the limit for $\varepsilon \rightarrow 0$
and study the convergence of the solution $[\theta_{\varepsilon},w_{\varepsilon}]$ of
the system \eqref{eq30}--\eqref{eq39}.
We present various convergence results based on the uniform estimates
 \eqref{est1_theta_a}--\eqref{est1_theta_b} and \eqref{est1_w_a}--\eqref{est1_w_b}
and, in addition,
prove some a priori estimates for the strong solution
$[\theta_{\varepsilon},w_{\varepsilon}]$ of the problem
\eqref{eq30}--\eqref{eq39}, which is equivalent to the system
\eqref{eq1d}--\eqref{eq3d} with the boundary and initial conditions \eqref{eq39c}--\eqref{eq39}.

First, recall that by Theorem \ref{th_sup_5} we have
\begin{equation}\label{est_30a}
\| \theta_{\varepsilon} \|_{L^2(I;W^{1,2}(\Omega))}, \; \|
\theta_{\varepsilon} \|_{L^{\infty}(I;L^2(\Omega))} \leq c,
\end{equation}
\begin{equation}\label{est_30b}
\| w_{\varepsilon} \|_{L^2(I;W^{1,2}(\Omega))}, \; \|
w_{\varepsilon} \|_{L^{\infty}(I;L^2(\Omega))} \leq c
\end{equation}
and from the equation \eqref{eq40} we derive that
\begin{equation}
\|(w_{\varepsilon})_t \|_{L^2(I;W^{1,2}(\Omega)^*)}  \leq c.
\end{equation}
As a consequence of the preceding a priori estimates we see that
there exist functions
 $w \in L^2(I;W^{1,2}(\Omega))$ and $\theta \in L^2(I;W^{1,2}(\Omega))$, $w_t \in
L^2(I;W^{1,2}(\Omega)^*)$, such that, along a selected
subsequence, we have ($\varepsilon_j\rightarrow 0^+$
 as $j\rightarrow \infty$)
\begin{eqnarray}
w_{\varepsilon_j} & \rightarrow & w  \qquad \textrm{ weakly in }
L^2(I;W^{1,2}(\Omega)),\label{conv_20}
\\
(w_{\varepsilon_j})_t & \rightarrow & w_t  \qquad \textrm{weakly in
} L^2(I;W^{1,2}(\Omega)^*),\label{conv_21}
\\
w_{\varepsilon_j} & \rightarrow & w  \qquad  \textrm{ almost
everywhere in }Q_T
\end{eqnarray}
and
\begin{equation}
\theta_{\varepsilon_j}  \rightarrow \theta \qquad \textrm{weakly in
} L^2(I;W^{1,2}(\Omega)).\label{conv_22}
\end{equation}
By the embedding $L^2(I;W^{1,2}(\Omega)) \cap
L^{\infty}(I;L^2(\Omega)) \hookrightarrow L^6(Q_T)$ we get the uniform bound
(using \eqref{est_30a} and \eqref{est_30b})
\begin{multline}
\|\theta_{\varepsilon}(\delta_w(\theta^{\varepsilon}_{\varepsilon},w^{\varepsilon}_{\varepsilon}) (w_{\varepsilon})_x
+
\delta_{\theta}(\theta^{\varepsilon}_{\varepsilon},w_{\varepsilon})(\theta^{\varepsilon}_{\varepsilon})_x )\|_{L^{3/2}(Q_T)}
\\
\leq
c_1 \|\theta_{\varepsilon}\|_{L^6(Q_T)}
\left(\|(w_{\varepsilon})_x\|_{L^2(Q_T)}
+ \|(\theta_{\varepsilon})_x\|_{L^2(Q_T)}\right)  \leq c_2
\end{multline}
and similarly
\begin{equation}
\|C_w w_{\varepsilon} \theta_{\varepsilon} +  \rho_S C_S \theta_{\varepsilon} \|_{L^{3/2}(0,T;W^{1,3/2}(\Omega))} \leq c.
\end{equation}
From the equation \eqref{eq2d} we obtain the uniform  estimate
\begin{equation}\label{est_50}
\| ( C_w w_{\varepsilon} \theta_{\varepsilon} +  \rho_S C_S \theta_{\varepsilon} )_t
\|_{L^{3/2}(I;W^{1,3}(\Omega)^*)}  \leq c.
\end{equation}
Since
\begin{displaymath}
W^{1,3/2}(\Omega) \hookrightarrow \hookrightarrow W^{1-\beta,3/2}(\Omega)
\hookrightarrow  W^{1,3}(\Omega)^*,
\end{displaymath}
where $\beta$ is a small positive real number, the Aubin-Lions lemma yields the existence
of $\chi \in L^{3/2}(I; W^{1-\beta,3/2}(\Omega))$ such that  (modulo a
subsequence)
\begin{equation}\label{conv_13}
C_w w_{\varepsilon_j} \theta_{\varepsilon_j} +  \rho_S C_S
\theta_{\varepsilon_j} \rightarrow \chi \qquad \textrm{strongly in }
L^{3/2}(I; W^{1-\beta,3/2}(\Omega)).
\end{equation}
Since \eqref{conv_13} yields the almost everywhere convergence and $w_{\varepsilon_j}$ converges
almost everywhere to $w$, we conclude
\begin{equation}\label{conv_14}
 \theta_{\varepsilon_j} \rightarrow \theta \qquad \textrm{almost everywhere in } Q_T.
\end{equation}
Hence, $C_w w_{\varepsilon_j} \theta_{\varepsilon_j} +  \rho_S C_S \theta_{\varepsilon_j} $
converges almost everywhere to $C_w w \theta +  \rho_S C_S \theta$ and
$\chi = C_w w \theta +  \rho_S C_S \theta$. Finally, \eqref{est_50} yields
\begin{equation}\label{conv_15}
(C_w w_{\varepsilon_j} \theta_{\varepsilon_j} +  \rho_S C_S
\theta_{\varepsilon_j})_t  \rightarrow (C_w w \theta +  \rho_S C_S
\theta)_t \qquad \textrm{weakly in } L^{3/2}(I;W^{1,3}(\Omega)^*).
\end{equation}


By \eqref{conv_20}--\eqref{conv_22} and \eqref{conv_14} we conclude
\begin{equation}\label{conv_40}
\lambda(\theta^{\varepsilon}_{\varepsilon_j},w^{\varepsilon}_{\varepsilon_j})
(\theta_{\varepsilon_j})_x \rightarrow \lambda(\theta,w) \theta_x
\qquad \textrm{weakly in } L^{2}(Q_T),
\end{equation}
\begin{equation}\label{conv_41}
\delta_w(\theta^{\varepsilon}_{\varepsilon_j},w^{\varepsilon}_{\varepsilon_j})
(w_{\varepsilon_j})_x  +
\delta_{\theta}(\theta^{\varepsilon}_{\varepsilon_j},w_{\varepsilon_j})
(\theta^{\varepsilon}_{\varepsilon_j})_x
\rightarrow
\delta_w(\theta,w)w_x + \delta_{\theta}(\theta,w)\theta_x \qquad
\textrm{weakly in } L^{2}(Q_T)
\end{equation}
and finally,
\begin{equation}\label{conv_42}
\theta_{\varepsilon_j}\left(
\delta_w(\theta^{\varepsilon}_{\varepsilon_j},w^{\varepsilon}_{\varepsilon_j})
(w_{\varepsilon_j})_x  +
\delta_{\theta}(\theta^{\varepsilon}_{\varepsilon_j},w_{\varepsilon_j})
(\theta^{\varepsilon}_{\varepsilon_j})_x \right)
\rightarrow
\theta\left(\delta_w(\theta,w)w_x +
\delta_{\theta}(\theta,w)\theta_x\right) \qquad \textrm{weakly in
}L^{3/2}(Q_T).
\end{equation}

Now let us present the convergence of the boundary conditions.
By \eqref{est1_w_a} and \eqref{conv_21} we deduce
that there exists a subsequence (not relabeled), such that
\begin{equation}\label{conv_16}
w_{\varepsilon_j}(0,\cdot) \rightarrow w(0,\cdot) \quad \textmd{ and
} \quad w_{\varepsilon_j}(\ell,\cdot) \rightarrow  w(\ell,\cdot)
\quad \textrm{strongly in }L^2(I) \textmd{ and almost everywhere in
} I.
\end{equation}

Taking fixed $x=0$ or $x=\ell$ and using \eqref{conv_13}
with $\chi = C_w w \theta +  \rho_S C_S \theta$ we conclude
\begin{eqnarray}
C_w w_{\varepsilon_j}(0,\cdot) \theta_{\varepsilon_j}(0,\cdot) +
\rho_S C_S \theta_{\varepsilon_j}(0,\cdot) &\rightarrow& C_w
w(0,\cdot)\theta(0,\cdot) + \rho_S C_S \theta(0,\cdot) \quad
\textrm{ almost everywhere in }I,
\nonumber\\
\label{conv_18a}
\\
C_w w_{\varepsilon_j}(\ell,\cdot) \theta_{\varepsilon_j}(\ell,\cdot)
+ \rho_S C_S \theta_{\varepsilon_j}(\ell,\cdot) &\rightarrow& C_w
w(\ell,\cdot)\theta(\ell,\cdot) + \rho_S C_S \theta(\ell,\cdot)
\quad \textrm{ almost everywhere in }I.
\nonumber\\
\label{conv_18b}
\end{eqnarray}

Now \eqref{conv_16}--\eqref{conv_18b} imply
\begin{equation}
\theta_{\varepsilon_j}(0,\cdot) \rightarrow \theta(0,\cdot) \quad \textmd{ and } \quad
\theta_{\varepsilon_j}(\ell,\cdot) \rightarrow  \theta(\ell,\cdot) \quad  \textrm{ almost everywhere in }I.
\label{conv_30}
\end{equation}

In order to get convergence results applicable to the radiative boundary conditions,
we need ``better'' uniform estimates than \eqref{est_30a}.
\begin{lemma}\label{bound_theta_2}
Let
$[\theta_{\varepsilon},w_{\varepsilon}] \in
\mathbf{V}^{2,1}_2(Q_T)$
be the solution of
the system \eqref{eq30}--\eqref{eq39}.
Then $\theta_{\varepsilon}$ is uniformly bounded in $L^{\infty}(I;L^4(\Omega))$, i.e.
\begin{equation}\label{est13_theta}
\|\theta_{\varepsilon} \|_{L^{\infty}(I;L^4(\Omega))}  \leq c,
\end{equation}
where the constant $c$ does not depend on $\varepsilon$.
\end{lemma}
\begin{proof}
Test \eqref{eq30} by $C_w \theta^4$ and \eqref{eq31} by $4\theta^3$ to
obtain (adding the both resulting equations)
\begin{align}\label{eq60}
& \int_{\Omega} \theta(x,t)^4 (C_w w(x,t) + \rho_S C_S ) {\rm d}{x}
+ \int_{Q_t}
12\theta^2\lambda({\theta}^{\varepsilon},{w}^{\varepsilon})\theta_x^2
{\rm d}x{\rm d}s+ \int_{Q_t} 3C_w \theta^4 d_s {\rm d}x{\rm d}s
\nonumber \\
&\quad + \int_{Q_t} 4 h_d \theta^3 d_s {\rm d}x{\rm d}s
+ \int_0^t
C_w \theta(\ell,s)^4
 \beta_c(\mathcal{P}(\theta^{\varepsilon}(\ell,s),w(\ell,s)) - \zeta P_{\infty})
 {\rm d}s
\nonumber \\
& \quad + \int_0^t 4\theta(\ell,s)^3 \left( (\alpha_c +
e\sigma|\theta^{\varepsilon}(\ell,s)|^3)\theta(\ell,s)
 - \zeta\vartheta^{\varepsilon}(s) \right) {\rm d}s
\nonumber \\
& =  \int_{\Omega} \theta(x,0)^4 (C_w w(x,0) + \rho_S C_S ) {\rm
d}{x}.
\end{align}
Simple calculation directly leads to the inequality
\begin{align}\label{eq60b}
& c_1 \int_{\Omega} \theta(x,t)^4 {\rm d}{x} + c_2 \int_{Q_t}
\theta^2\theta_x^2 {\rm d}x{\rm d}s + c_3 \int_0^t \theta(\ell,s)^4
{\rm d}s
\nonumber \\
& \leq \int_{\Omega} \theta(x,0)^4 (C_w w(x,0) + \rho_S C_S ) {\rm
d}{x} - \int_{Q_t} 4 h_d \theta^3 d_s {\rm d}x{\rm d}s
\nonumber \\
& \qquad + \int_0^t C_w \theta(\ell,s)^4 \beta_c  P_{\infty} {\rm
d}s + \int_0^t 4\theta(\ell,s)^3 \vartheta^{\varepsilon}(s)  {\rm
d}s.
\end{align}
The integrals on the right hand side of \eqref{eq60b} can be estimated using Young's inequality
and \cite[Remark 4]{FiloKacur1995}, respectively, in the following manner:
\begin{equation}\label{sup_11}
\int_{Q_t} 4 h_d \theta^3 d_s {\rm d}x{\rm d}s
\leq
3 h_d \int_{Q_t}
\theta^4 {\rm d}x{\rm d}s + h_d \int_{Q_t} d_s^4(x,s) {\rm d}x{\rm d}s,
\end{equation}
further,
\begin{equation}\label{sup_12}
\int_0^t C_w \theta(\ell,s)^4 \beta_c  P_{\infty} {\rm d}s
\leq
C_w \beta_c  P_{\infty} \left( \int_{Q_t}  \epsilon 4\theta^2 \theta_x^2 +
C(\epsilon)\theta^4 {\rm d}x{\rm d}s \right)
\end{equation}
and finally,
\begin{equation}\label{sup_13}
\int_0^t 4\theta(\ell,s)^3 \vartheta^{\varepsilon}(s)  {\rm d}s
\leq
\epsilon \int_0^t \theta^4(\ell,s){\rm d}s + C(\epsilon)\int_0^t
\vartheta^{\varepsilon}(s)^4{\rm d}s,
\end{equation}
where $\epsilon$ represents sufficiently small positive real number.
Now, taking \eqref{eq60b}--\eqref{sup_13} together, we obtain
\begin{equation}\label{eq65c}
c_1 \int_{\Omega} \theta(x,t)^4 {\rm d}{x}
\leq
c_2 + c_3 \int_{Q_t} \theta^4{\rm d}x{\rm d}s,
\end{equation}
which yields, applying the Gronwall's inequality, the uniform estimate \eqref{est13_theta}.
The proof is complete.

\end{proof}

By the Sobolev embedding theorem (see \cite[Theorem 4.12]{AdamsFournier1992}
or \cite[Theorem 8.1.2]{KufFucJoh1977}) we have
\begin{equation}\label{emb_20}
W^{1,2}(\Omega) \hookrightarrow W^{3/4,4}(\Omega).
\end{equation}
Raising and integrating the
interpolation inequality \cite[Theorem 5.2]{AdamsFournier1992}
($\epsilon$ means small positive real number)
\begin{equation}\label{int_ineq_010}
\|\theta_{\varepsilon}\|_{W^{1/4+\epsilon,4}(\Omega)}
\leq c
\|\theta_{\varepsilon}\|^{(1+4\epsilon)/3}_{W^{3/4,4}(\Omega)}
\|\theta_{\varepsilon}\|^{(2-4\epsilon)/3}_{L^4(\Omega)}
\end{equation}
from $0$ to $T$ we get
\begin{eqnarray}\label{est_121}
 \left( \int^T_0 \|\theta_{\varepsilon}\|^{6/(1+4\epsilon)}_{W^{1/4+\epsilon,4}(\Omega)}
 {\rm d}t \right)^{(1+4\epsilon)/6}
 &\leq&
 c
 \left( \int^T_0 \|\theta_{\varepsilon}\|^2_{W^{3/4,4}(\om)}
\|\theta_{\varepsilon}\|^{4(1-2\epsilon)/(1+4\epsilon)}_{L^4(\Omega)} {\rm d}t
\right)^{(1+4\epsilon)/6}
\nonumber\\
&\leq&
c
\| \theta_{\varepsilon} \|^{(1+4\epsilon)/3}_{{L}^{2}(I;W^{3/4,4}(\Omega))}
\| \theta_{\varepsilon} \|^{(2-4\epsilon)/3}_{L^{\infty}(I;L^4(\om))}.
\end{eqnarray}
Now taking into account \eqref{est_30a}, \eqref{est13_theta},
\eqref{emb_20} and \eqref{est_121} we arrive at the estimate
\begin{equation}
\|\theta_{\varepsilon}\|_{L^{6/(1+4\epsilon)}(I;W^{1/4+\epsilon,4}(\Omega))}
\leq c.
\end{equation}
Consequently, we get (along a selected subsequence)
\begin{equation}\label{conv_50}
\theta_{\varepsilon_j}(0,\cdot) \rightarrow \theta(0,\cdot) \quad
\textmd{ and } \quad \theta_{\varepsilon_j}(\ell,\cdot) \rightarrow
\theta(\ell,\cdot) \qquad \textrm{weakly in }L^p(I),\; 1 \leq p < 6.
\end{equation}

The strong solution $[\theta_{\varepsilon},w_{\varepsilon}] \in
\mathbf{V}^{2,1}_2(Q_T)$ of the problem \eqref{eq30}--\eqref{eq39}
 (ensured by Theorem \ref{th_sup_5}) is a solution of the system
\eqref{eq1d}--\eqref{eq3d} with the boundary and initial conditions \eqref{eq39c}--\eqref{eq39}
and satisfies the variational problem
(corresponding to \eqref{eq1d}--\eqref{eq3d} and \eqref{eq39c}--\eqref{eq39})

\begin{align}\label{weak_form_limit}
& -\int_{Q_T} \left( w_{\varepsilon} - d_{\varepsilon} \right)
\phi_t +\left( C_w \theta_{\varepsilon} w_{\varepsilon} + \rho_S C_S
\theta_{\varepsilon} + h_d d_{\varepsilon} \right) \psi_t {\rm
d}x{\rm d}t
\nonumber \\
& \qquad +\int_{Q_T} \left(
\delta_w(\theta^{\varepsilon}_{\varepsilon},w^{\varepsilon}_{\varepsilon})
(w_{\varepsilon})_x  +
\delta_{\theta}(\theta^{\varepsilon}_{\varepsilon},w_{\varepsilon})
(\theta^{\varepsilon}_{\varepsilon})_x \right) \phi_x +
\lambda(\theta^{\varepsilon}_{\varepsilon},w^{\varepsilon}_{\varepsilon})
(\theta_{\varepsilon})_x \psi_x {\rm d}x{\rm d}t
\nonumber \\
& \qquad -\int_{Q_T} \theta_{\varepsilon}\left(
\delta_w(\theta^{\varepsilon}_{\varepsilon},w^{\varepsilon}_{\varepsilon})
(w_{\varepsilon})_x  +
\delta_{\theta}(\theta^{\varepsilon}_{\varepsilon},w_{\varepsilon})
(\theta^{\varepsilon}_{\varepsilon})_x \right) \psi_x \;{\rm d}x{\rm
d}t
\nonumber  \\
& \qquad +\int_0^T \beta_c
\left(\mathcal{P}(\theta^{\varepsilon}_{\varepsilon}(\ell,t),w_{\varepsilon}(\ell,t))
- P_{\infty}\right) \phi(\ell,t) \;{\rm d}t
\nonumber \\
& \qquad + \int_0^T
\left[(\alpha_c+e\sigma|\theta^{\varepsilon}_{\varepsilon}(\ell,t)|^3)\theta_{\varepsilon}(\ell,t)
 - \vartheta^{\varepsilon}(t)\right]\psi(\ell,t)  {\rm d}t
\nonumber \\
& \qquad + \int_0^T  C_w \theta_{\varepsilon}(\ell,t) \beta_c
 \left(\mathcal{P}(\theta^{\varepsilon}_{\varepsilon}(\ell,t),
 w_{\varepsilon}(\ell,t)) - P_{\infty}\right)\psi(\ell,t)
{\rm d}t
\nonumber \\
& = \int_{\Omega} w^{\varepsilon}_0\phi(x,0) + \left( C_w
\theta^{\varepsilon}_0 w^{\varepsilon}_0 + \rho_S C_S
\theta^{\varepsilon}_0 \right)\psi(x,0) {\rm d}x
\end{align}
for all test functions $[\phi,\psi]\in
C^{\infty}(\overline{Q}_T)$, $\phi(x,T)=\psi(x,T)=0$ $\forall x \in
\Omega$.

The above established convergences \eqref{conv_20}--\eqref{conv_22},
\eqref{conv_15}--\eqref{conv_16} and \eqref{conv_50}
are sufficient for taking the limit $\varepsilon_j \rightarrow 0$
as $j \rightarrow \infty$ in \eqref{weak_form_limit}
(along a selected subsequence) to get the weak solution
of the system \eqref{eq1b}--\eqref{eq10b} satisfying \eqref{weak_form}.
This completes the proof of the main result stated by Theorem \ref{main_result}.


\section{Illustration of transport processes
in concrete walls at elevated temperatures}\label{sec:numerical_exp}

\subsection{Illustrative example}

 We assume a concrete wall with a width of $120~{\rm mm}$
 subjected to fire on one side while on the other side,
 the wall is assumed to be thermal and moisture insulated,
 see Fig.~\ref{Anlysed_Wall}. Three different fire scenarios,
 which represent the time dependency of the ambient temperature
 on the exposed side, are employed: (i) standard fire curve (ISO fire),
 (ii) hydrocarbon fire curve (HC fire), and (iii) parametric fire curve (PM fire),
 see Fig.~\ref{Fire_Curves}. The last one is more sophisticated than the others
 since for the PM fire, the temperature is dependent not only
 on the time of fire but also on the specific parameters of
 a given fire compartment (in our case, the following
 parameters are assumed: $q_{td} = 160~{\rm{MJ\,m^{-2}}}$,
 $O = 0.12~{\rm{m^{1/2}}}$, $b = 1000~{\rm{J\,m^{-2}\,s^{-1/2}\,K^{-1}}}$,
 fire growth rate: medium, see~\cite{Eurocode1}). As shown in Fig.~\ref{Fire_Curves},
 this curve also includes a decreasing branch that simulates a cooling phase of a fire.
 More information about the fire scenarios mentioned above can be found in~\cite{Eurocode1}.

\begin{figure}[h]
\begin{minipage}[b]{.5\textwidth}
\includegraphics[angle=0,width=5.2cm]{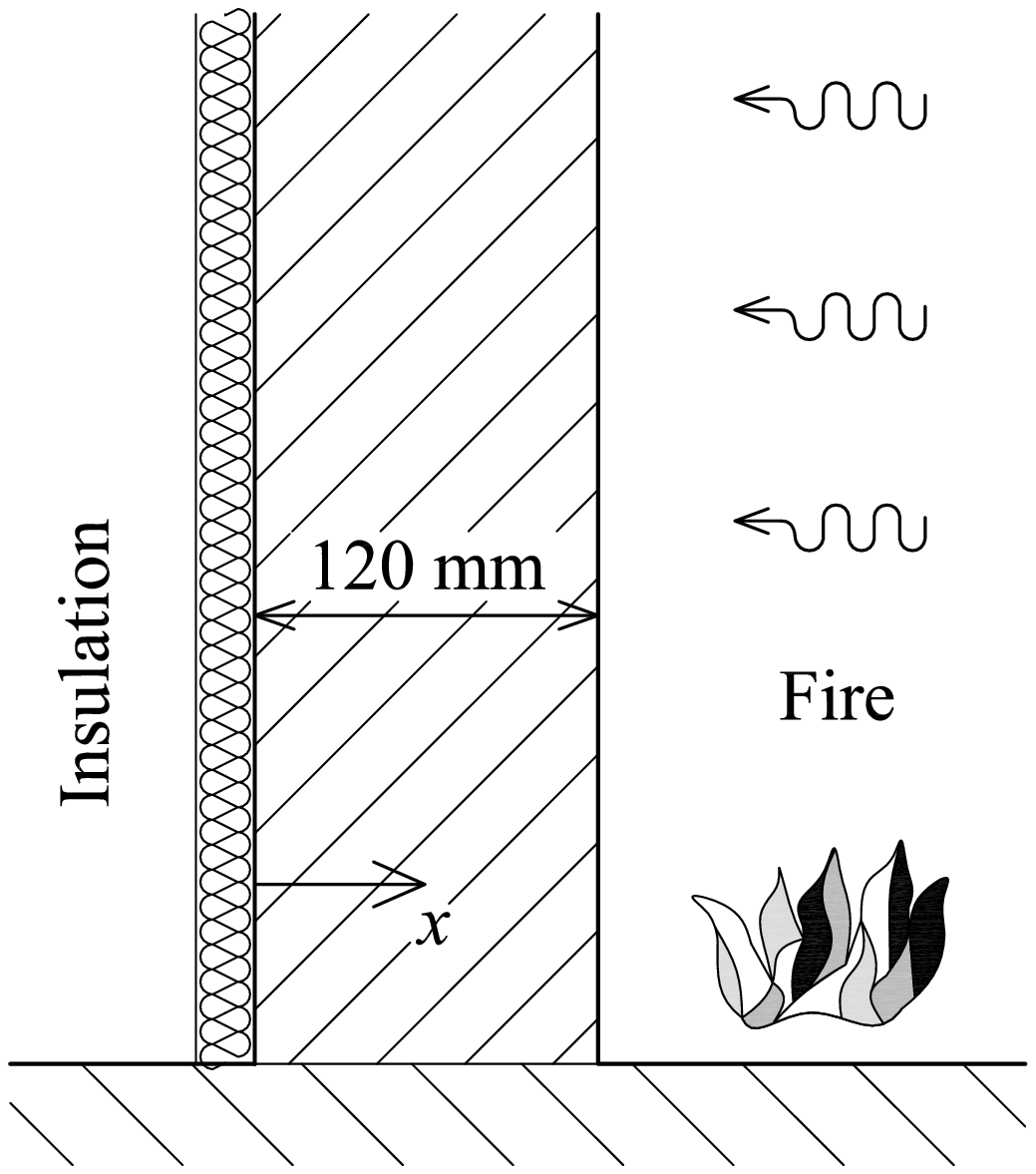}
\caption{Analyzed wall.}
\label{Anlysed_Wall}
\end{minipage}
\begin{minipage}[b]{.5\textwidth}
\includegraphics[angle=0,width=8cm]{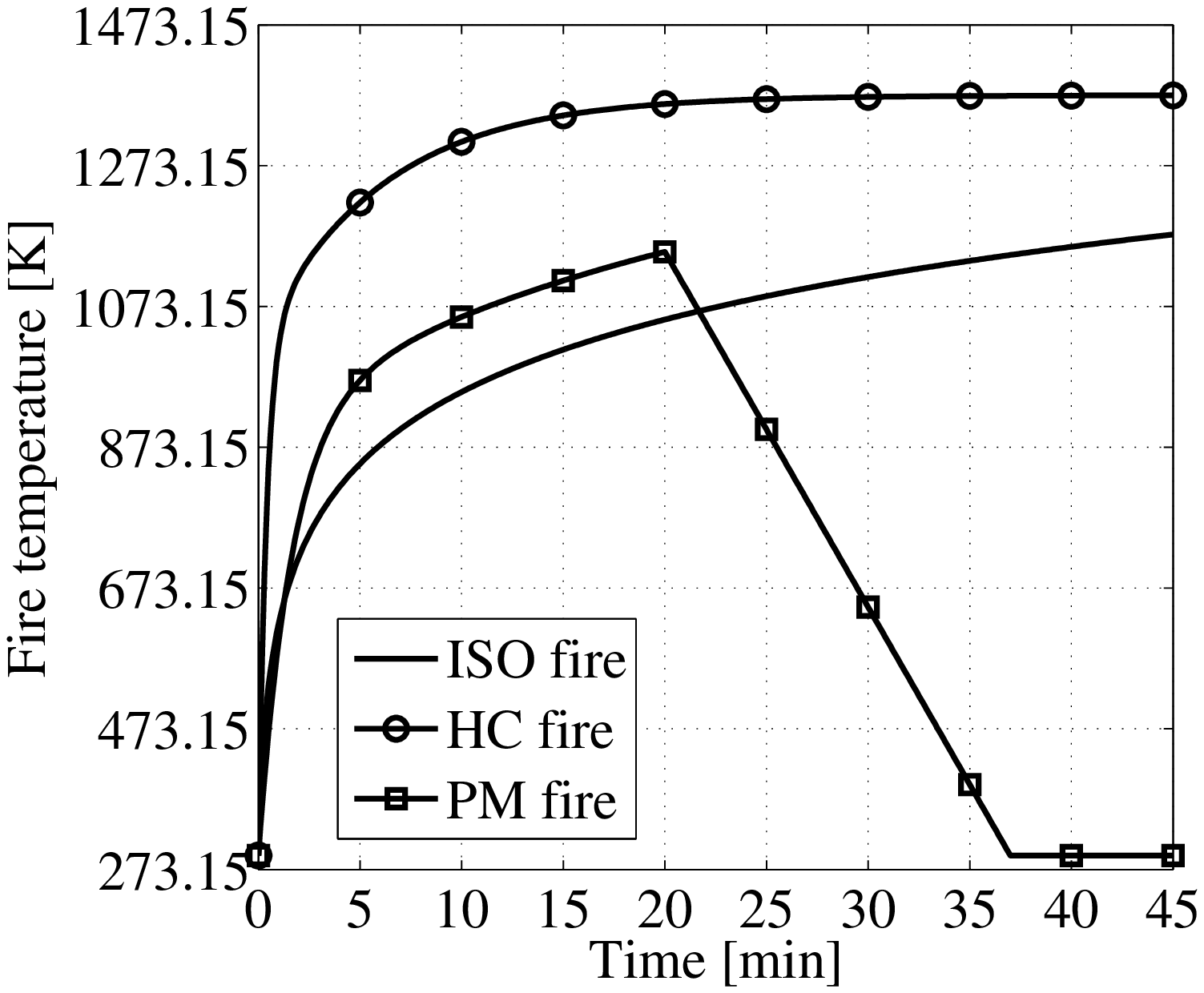}
\caption{Fire curves.}
\label{Fire_Curves}
\end{minipage}
\end{figure}

 On the exposed side of the analyzed wall, the constant parameters
 of the boundary conditions are set to $\beta_c = 0.019~{\rm{m\,s^{-1}}}$,
 $P_{\infty} = 1.7542 \times 10^3~{\rm{Pa}}$, $e = 0.7$,
 $\sigma = 5.67 \times 10^{-8}~{\rm{W\,m^{-2}\,K^{-4}}}$,
 and $\alpha_{c} = 25,~50,~{\rm{or}}~35~{\rm{W\,m^{-2}\,K^{-1}}}$
 for ISO fire, HC fire or PM fire, respectively.
 The temperature $\theta_{\infty}(t)$ is assumed
 to be equal to the fire temperature according to Fig.~\ref{Fire_Curves}.

 The uniform initial conditions $\theta_0 = 293.15~{\rm{K}}$
 and $w_0 = 71.01~{\rm{kg\,m^{-3}}}$ are considered.

\subsubsection{Material data}
The material properties of concrete at high temperatures are assumed as follows.

The thermal conductivity of concrete $\Lambda = \Lambda(\theta,P)$
is adopted from~\cite[(46)--(47)]{GaMaSch1999}, where $\Lambda_{d0}
= 1.3863~{\rm{W\,m^{-1}\,K^{-1}}}$, $A_{\Lambda} =
-0.0007272~{\rm{K^{-1}}}$ (see~\cite[Tab.~2]{GaPeSch2011}). The
concrete porosity $n = n(\theta)$ is assumed according
to~\cite[(A7)]{TenLiPur2001}, with $n_0 = 0.1$, and the saturation
$S = S(\theta,P)$ according to~\cite[(10)]{ChungCon2005}. The
permeability of concrete $\kappa = \kappa(\theta,P) = a(\theta,P)/g$
is taken from~\cite[(12a)--(12b)]{BaTh1978}, where $a_0 =
10^{-13}~{\rm{m\,s^{-1}}}$. The mass of dehydrated water is
determined by Eq.~\eqref{eq1b_dehydr}, which is adopted
from~\cite[(M-4)]{PontEhrlacher}). Here, the term $d_{eq}(\theta)$
is taken from ~\cite[(1.9)]{Alnajim2004}) with
$d_{eq}^{378.15~{\rm{K}}} = 330~{\rm{kg\,m^{-3}}}$ and $\tau =
10800~{\rm{s}}$. The sorption isotherm functions are assumed
according to~\cite{BaTh1978} with some modification based
on~\cite[(73)]{Davie:2010:FG} and~\cite[(29)]{DwaiKod2009}. Here,
the mass of anhydrous cement per unit volume of concrete $c =
250~{\rm{kg\,m^{-3}}}$ and the saturation water content at the room
temperature $w_{0s} = 100~{\rm{kg\,m^{-3}}}$. The parameters $h_d$,
$\rho_S$, $C_S$ and $C_w$ are considered to be constant values,
namely $h_d = 2.5 \times 10^6~{\rm{J\,kg^{-1}}}$, $\rho_S =
2400~{\rm{kg\,m^{-3}}}$, $C_S = 900~{\rm{J\,kg^{-1}\,K^{-1}}}$ and
$C_w = 2080~{\rm{J\,kg^{-1}\,K^{-1}}}$.

\subsubsection{Numerical procedure}
In order to obtain an approximate solution of the nonlinear model,
the well-known Galerkin procedure can be employed.
The spacial discretization is performed by the
one-dimensional finite element method.
We consider linear elements with the element
size of $0.0005~{\rm{m}}$ (240 elements in total).
The time discretization is carried out by a semi-implicit
difference scheme. In our case, we assume the time step $\Delta t = 0.5~{\rm{s}}$.
The numerical procedure is described in detail in~\cite{Benes2011}.

An algorithm arising from the numerical scheme mentioned above
has been included in an in-house MATLAB code, which is
employed to determine the distribution of temperature,
pore pressure and water content in the analyzed wall.

\subsubsection{Results}

\begin{figure}
\begin{minipage}[b]{.5\textwidth}
\begin{center}
\includegraphics[angle=0,width=7.0cm]{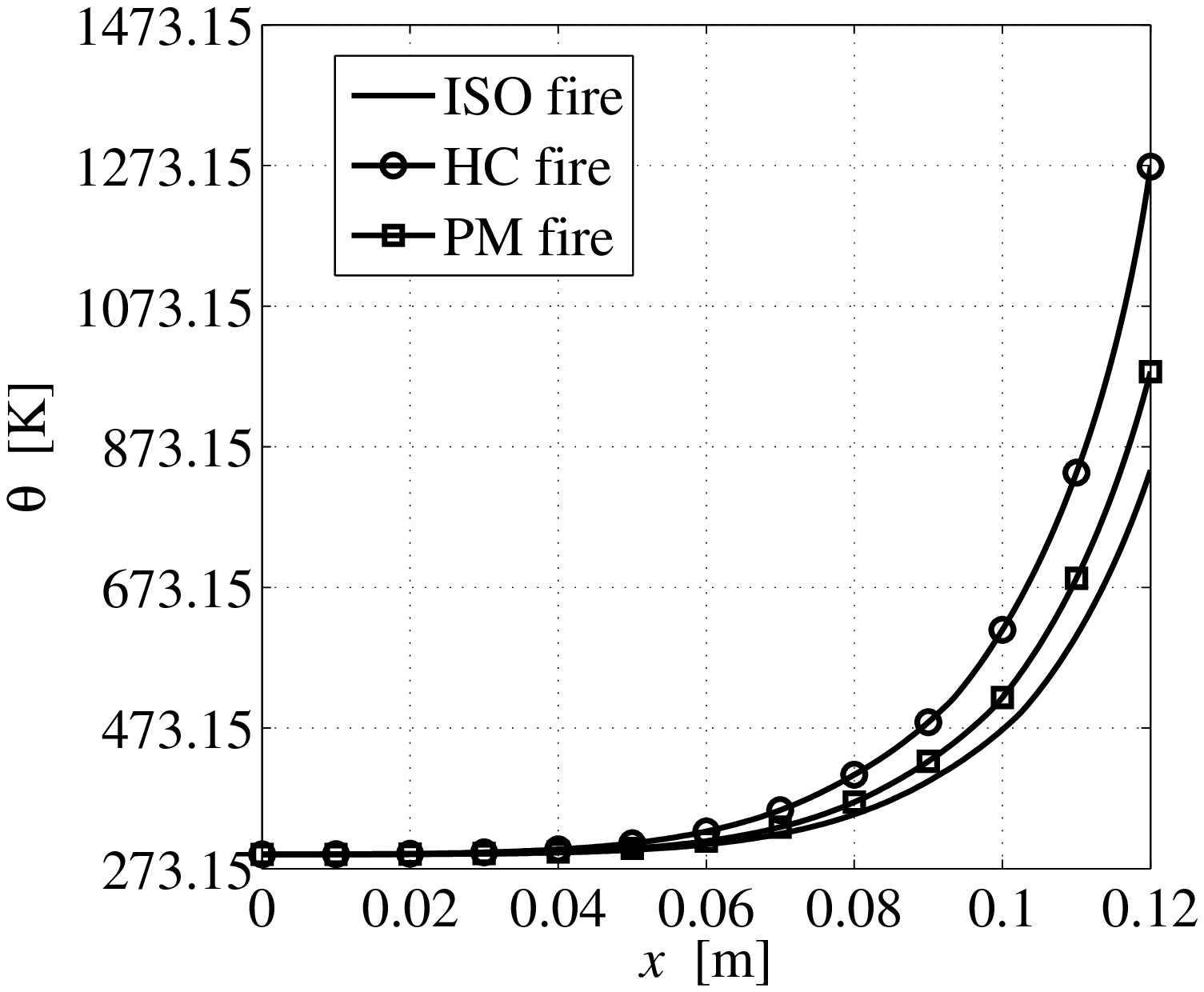}
\end{center}
\end{minipage}
\begin{minipage}[b]{.5\textwidth}
\begin{center}
\includegraphics[angle=0,width=7.0cm]{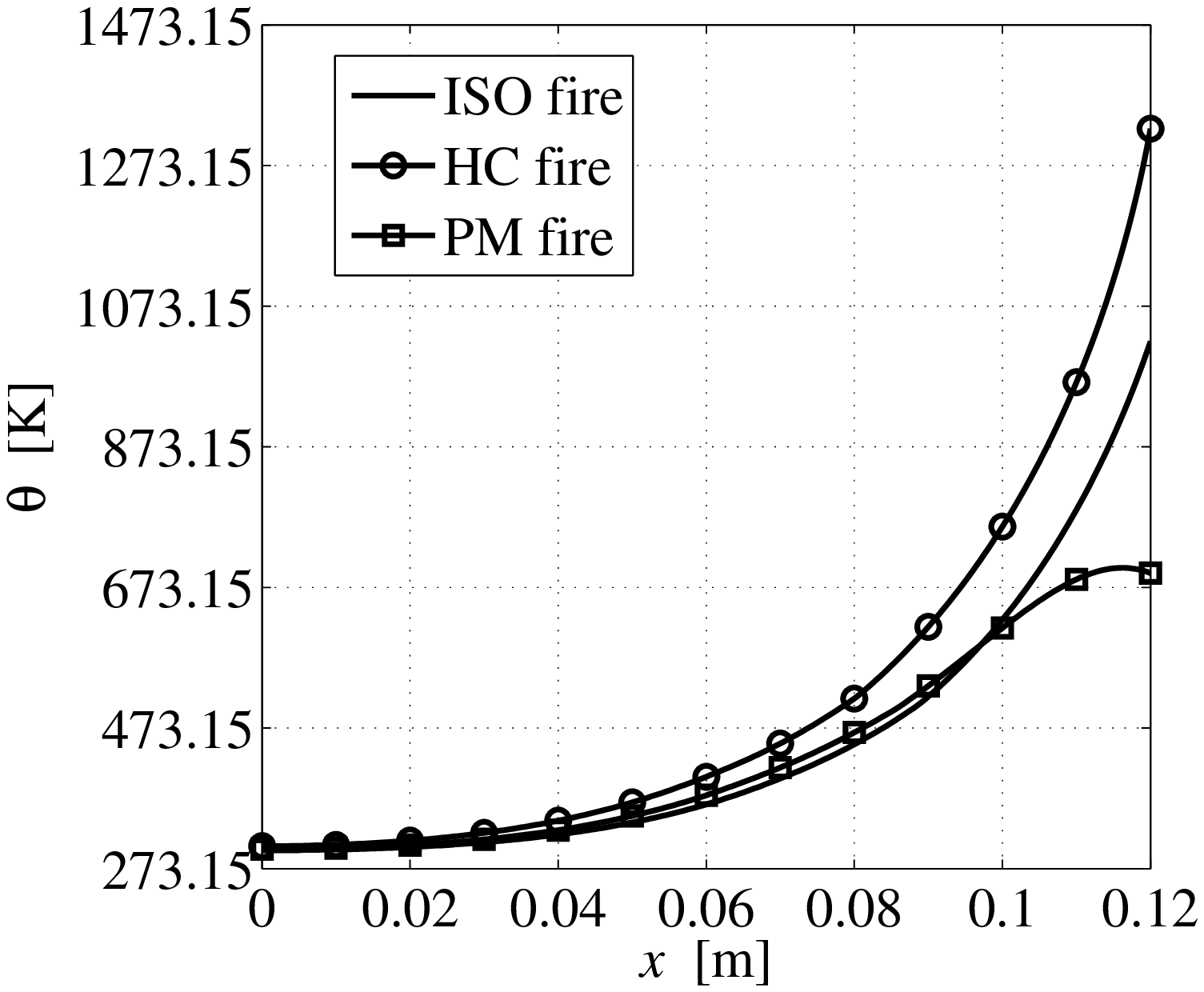}
\end{center}
\end{minipage}
\begin{minipage}[b]{.5\textwidth}
\begin{center}
\includegraphics[angle=0,width=7.0cm]{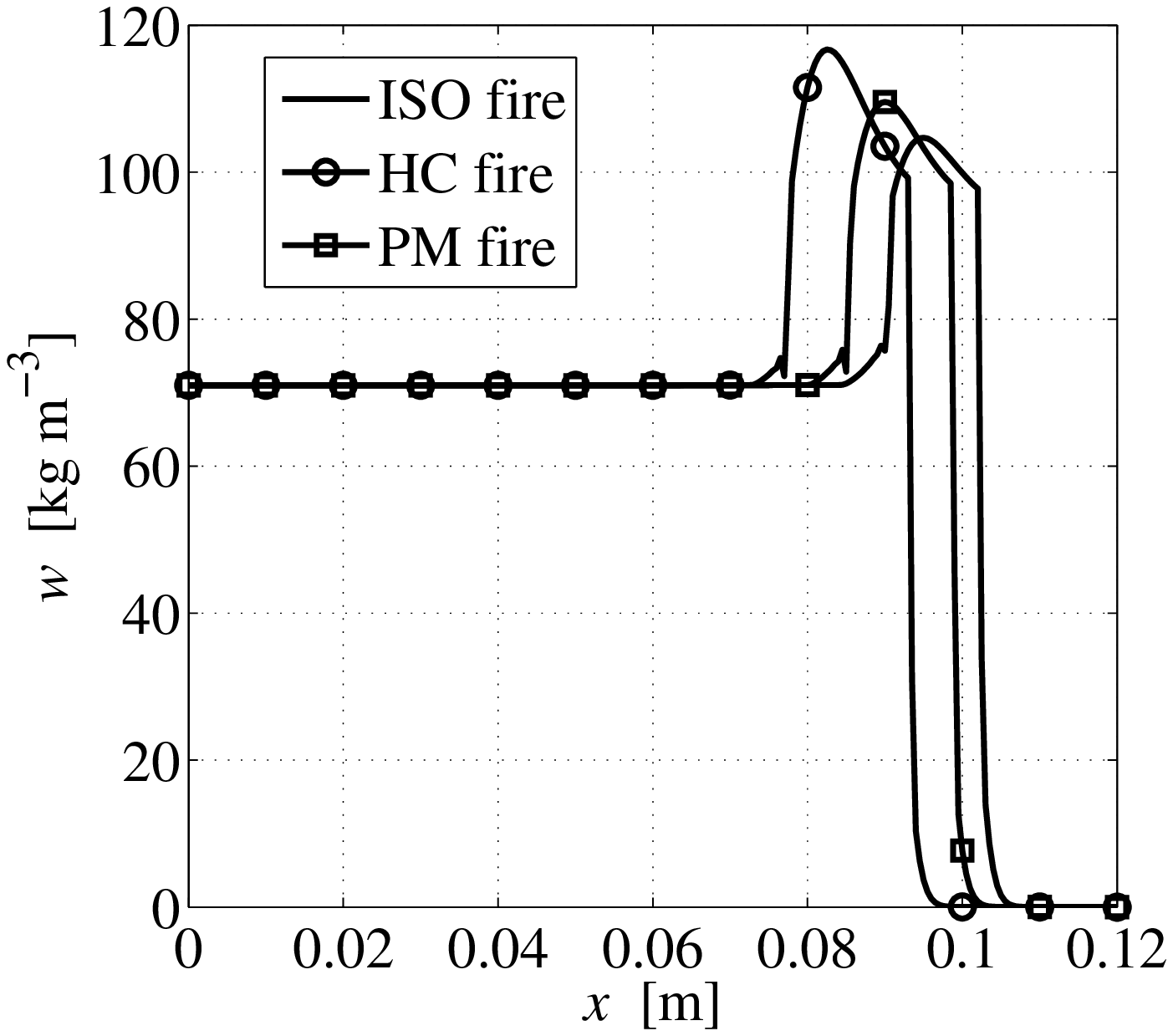}
\end{center}
\end{minipage}
\begin{minipage}[b]{.5\textwidth}
\begin{center}
\includegraphics[angle=0,width=7.0cm]{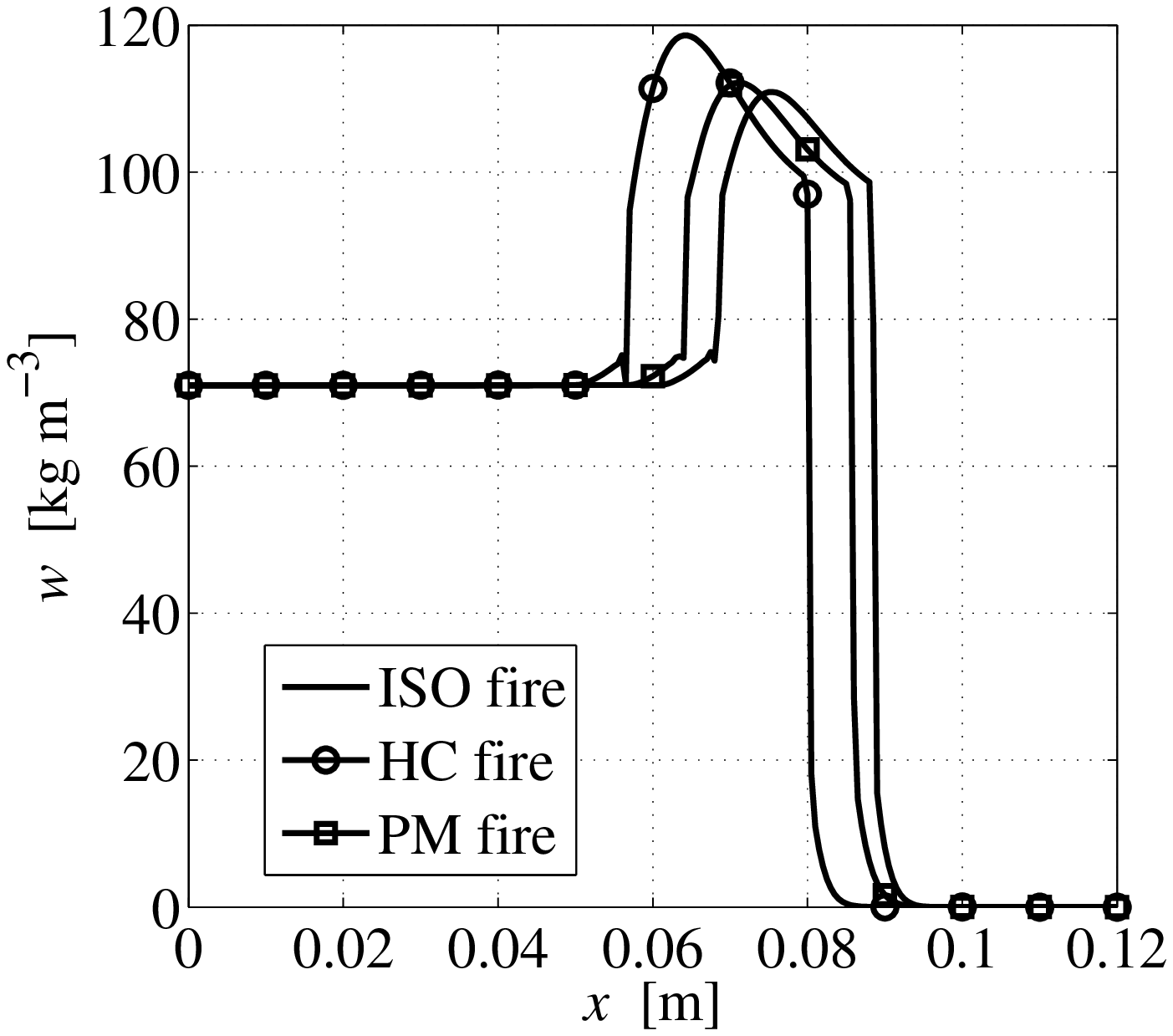}
\end{center}
\end{minipage}
\begin{minipage}[b]{.5\textwidth}
\begin{center}
\includegraphics[angle=0,width=7.0cm]{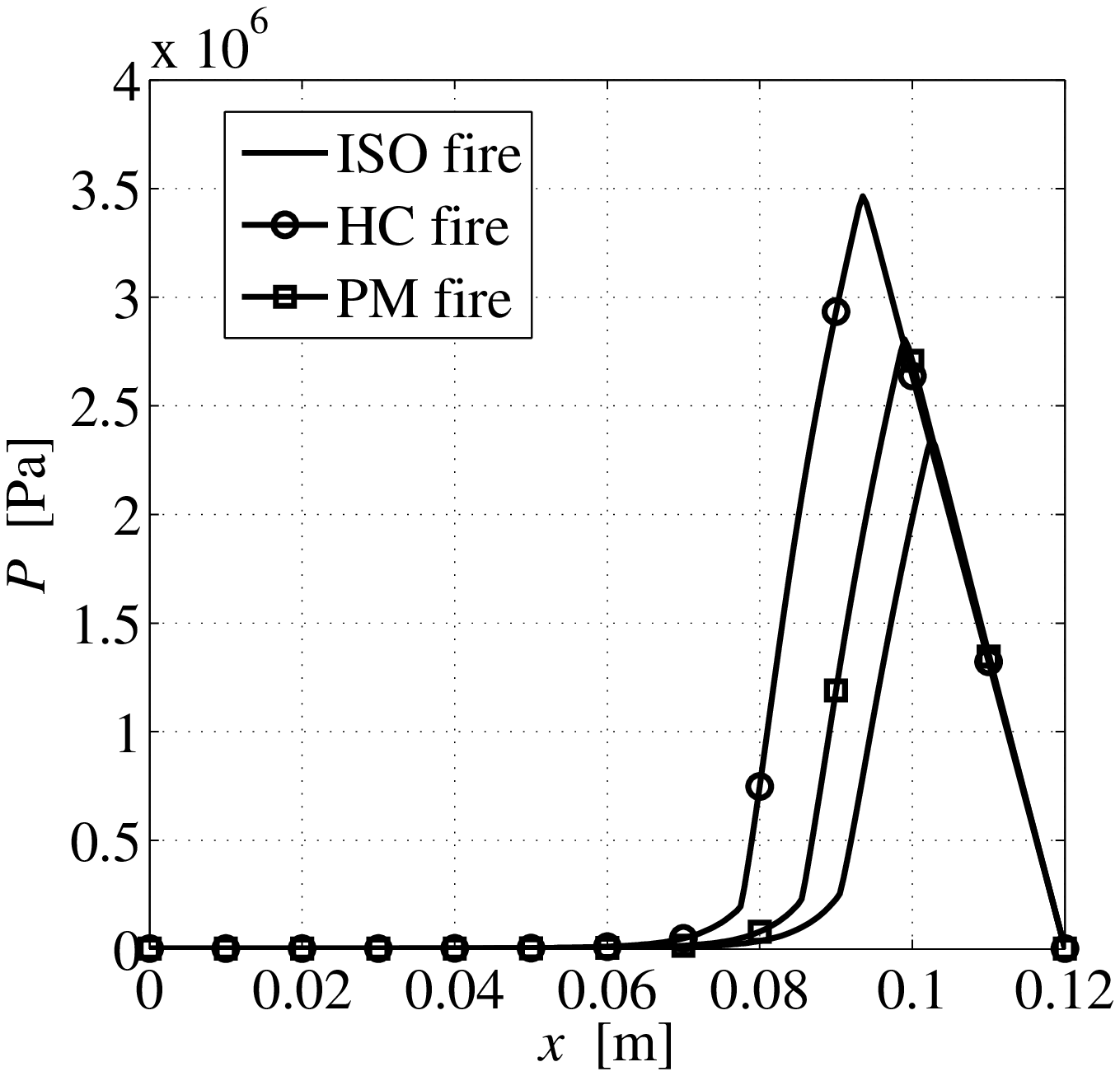}
\end{center}
\end{minipage}
\begin{minipage}[b]{.5\textwidth}
\begin{center}
\includegraphics[angle=0,width=7.0cm]{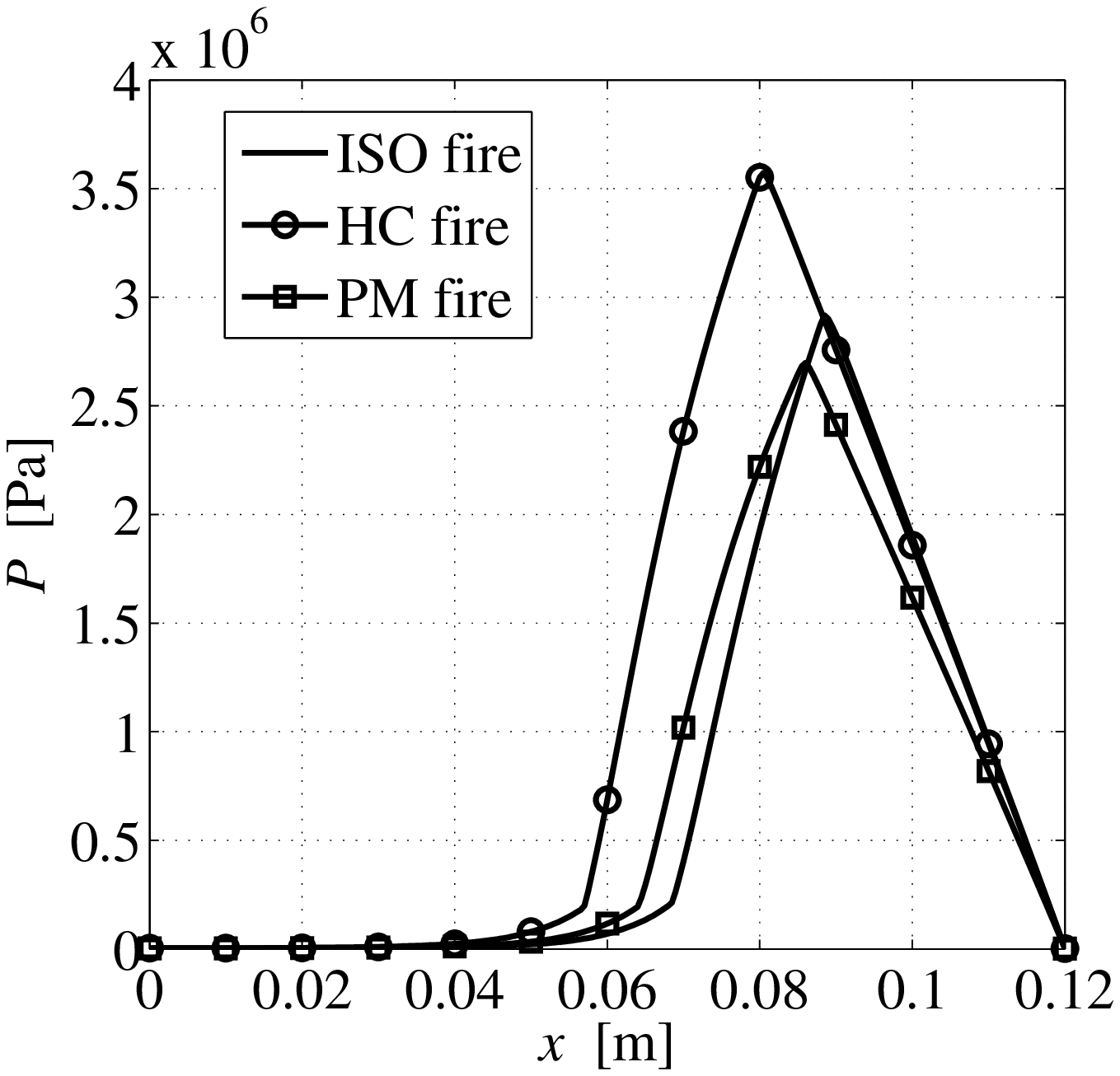}
\end{center}
\end{minipage}
\caption{Temperature, water content and pore pressure
distribution across the analyzed wall: fire exposure of
15 minutes (left) and 30 minutes (right).}\label{1D_Results}
\end{figure}

\begin{figure}
\begin{minipage}[b]{.5\textwidth}
\begin{center}
\includegraphics[angle=0,width=7.0cm]{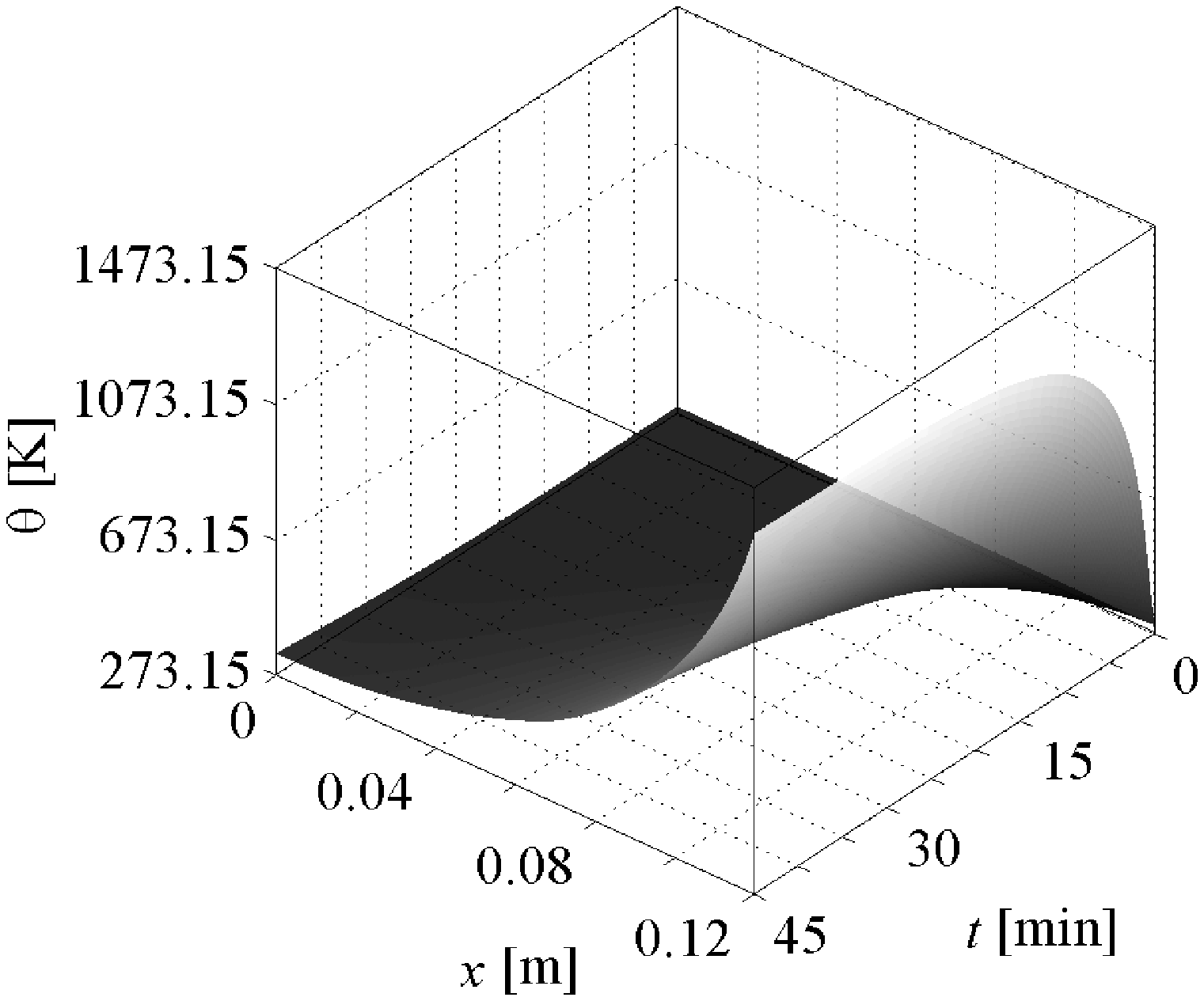}
\end{center}
\end{minipage}
\begin{minipage}[b]{.5\textwidth}
\begin{center}
\includegraphics[angle=0,width=7.0cm]{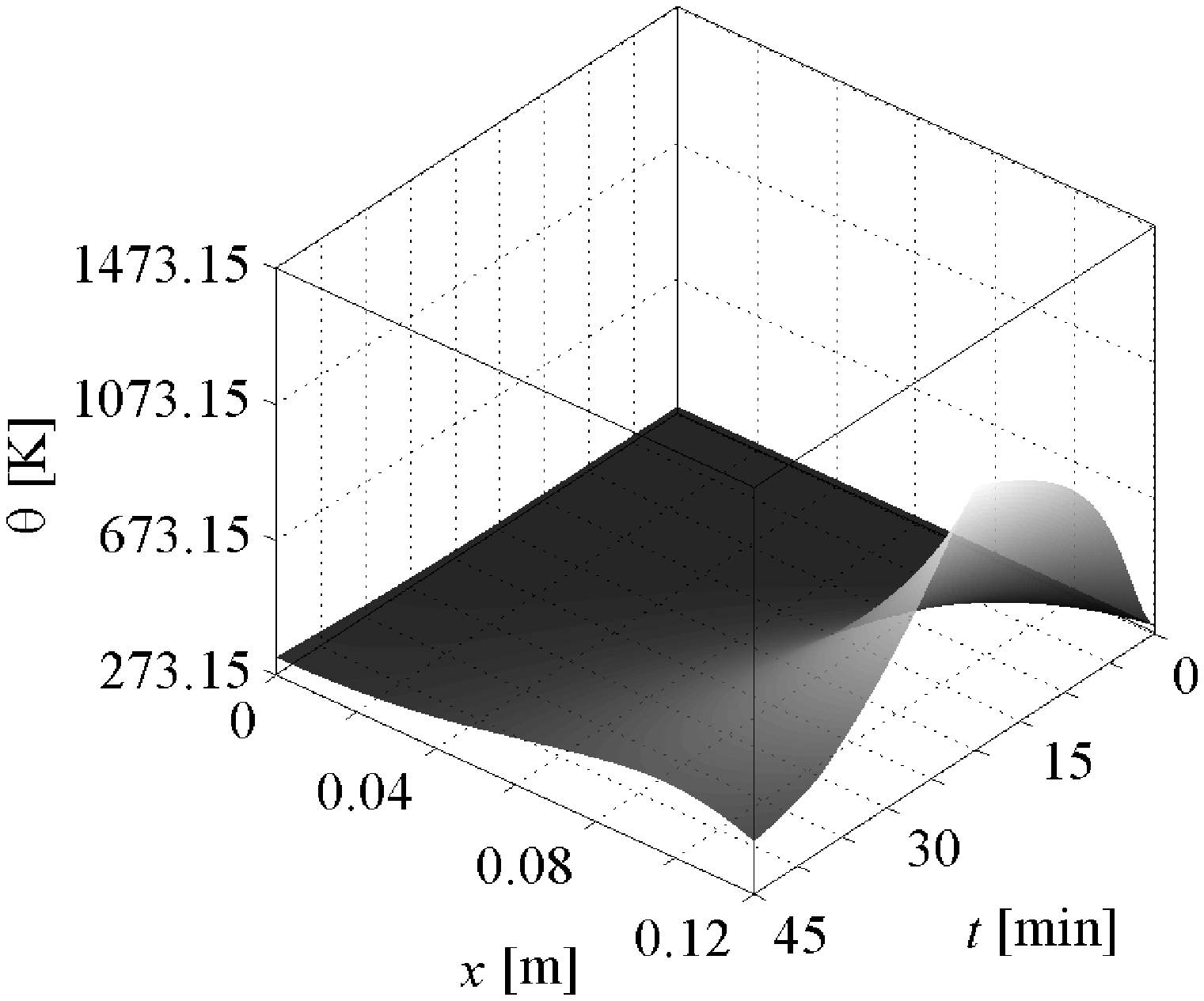}
\end{center}
\end{minipage}

\begin{minipage}[b]{.5\textwidth}
\begin{center}
\includegraphics[angle=0,width=7.0cm]{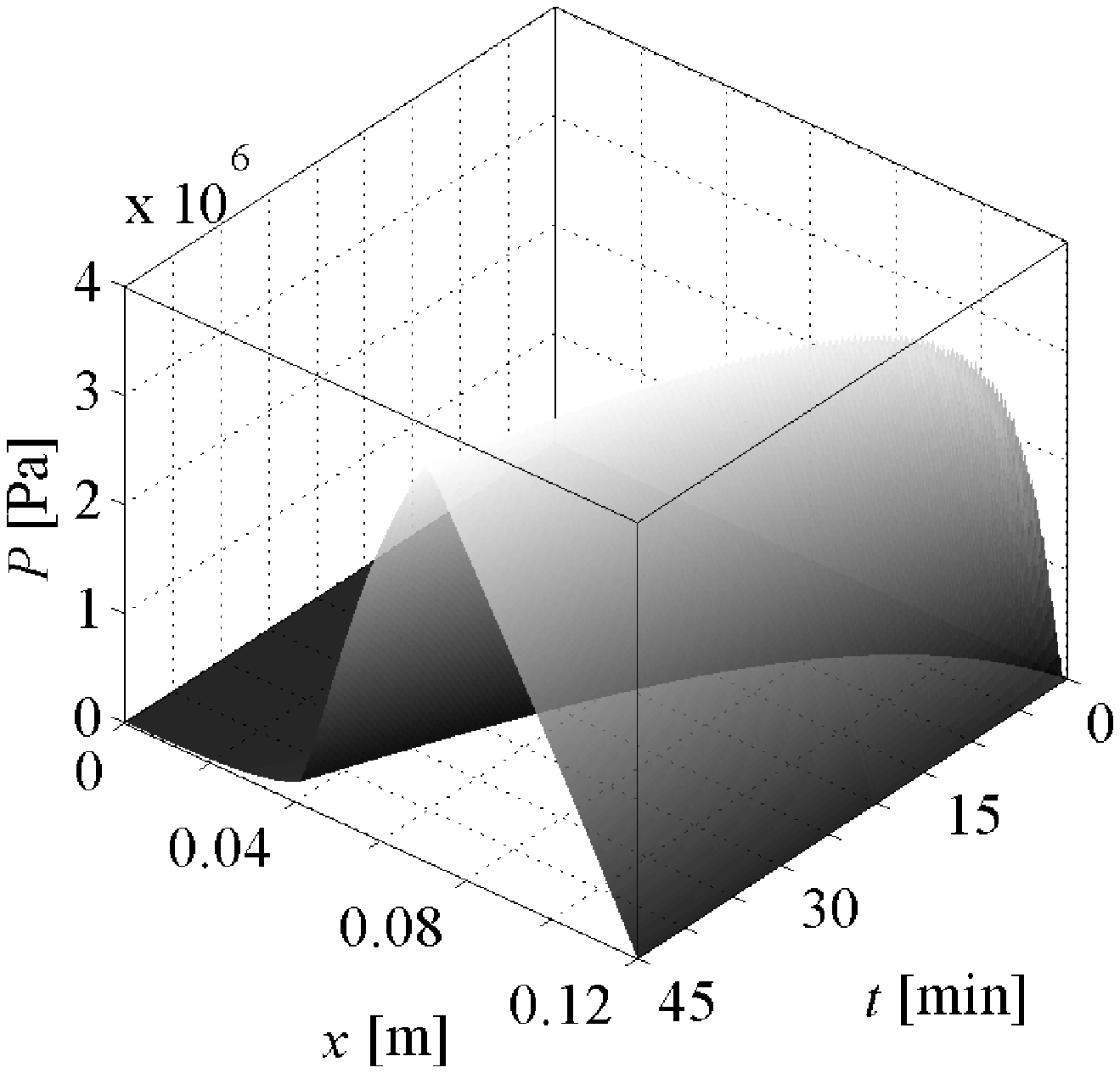}
\end{center}
\end{minipage}
\begin{minipage}[b]{.5\textwidth}
\begin{center}
\includegraphics[angle=0,width=7.0cm]{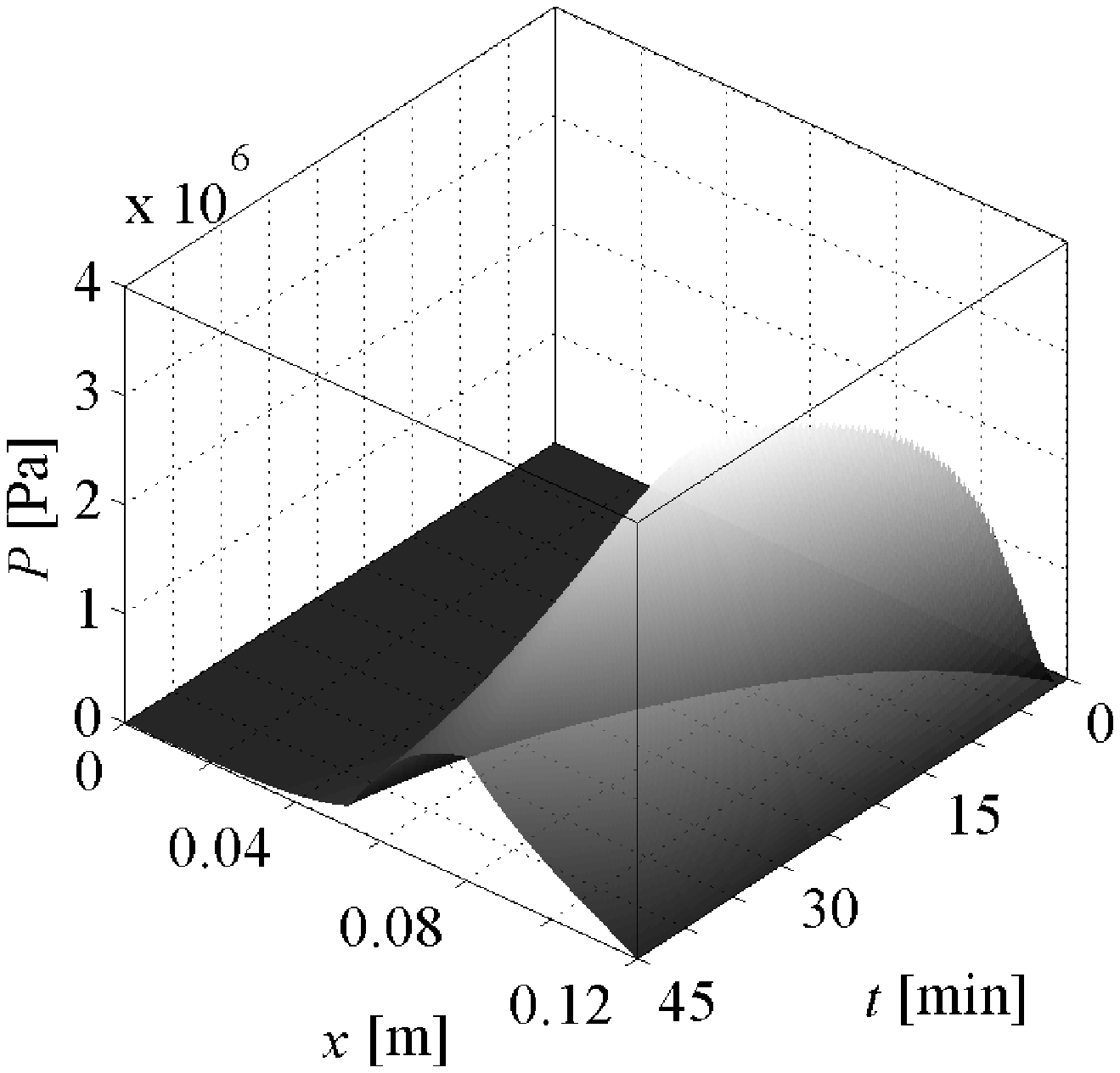}
\end{center}
\end{minipage}

\begin{minipage}[b]{.5\textwidth}
\begin{center}
\includegraphics[angle=0,width=7.0cm]{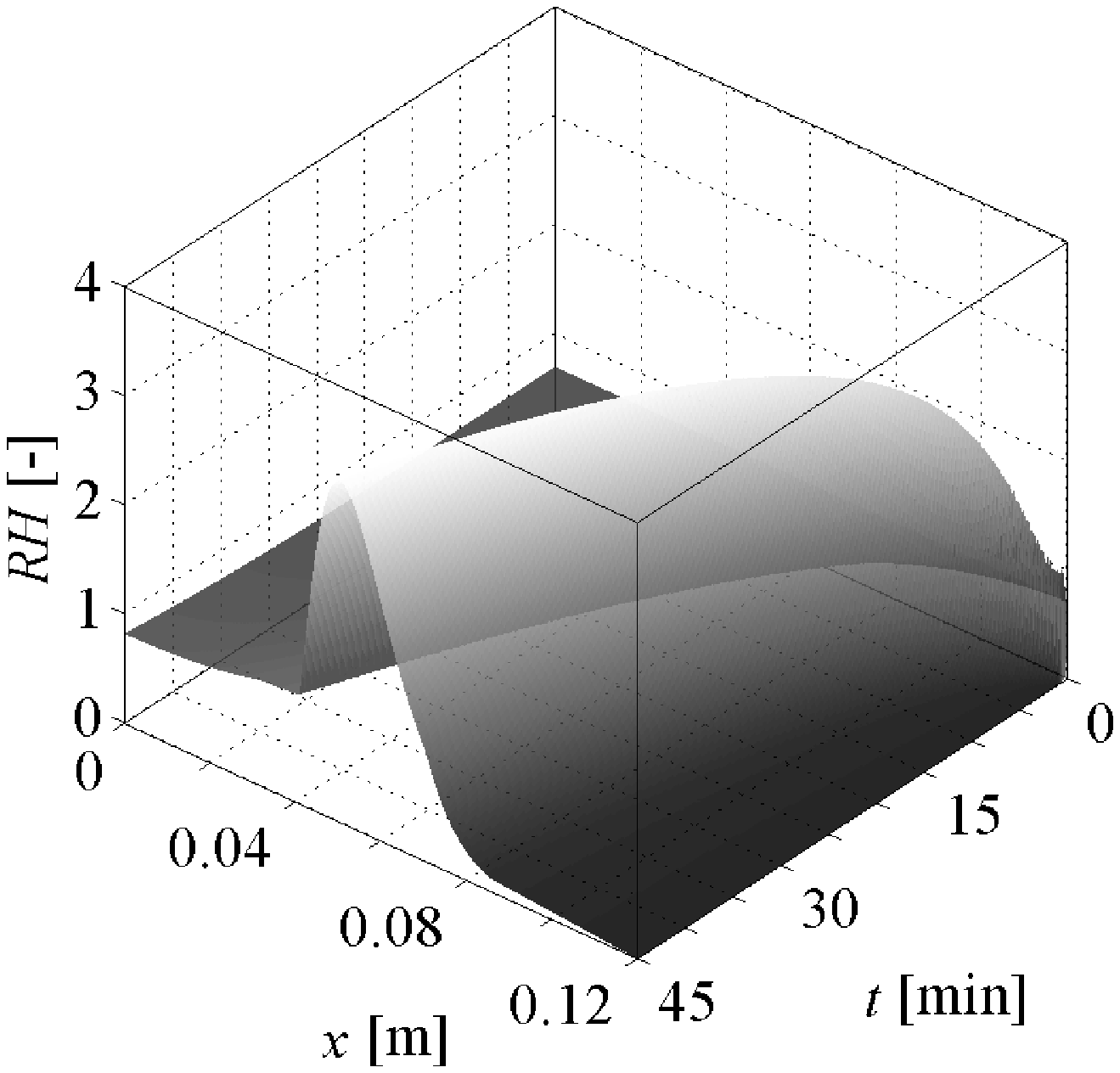}
\end{center}
\end{minipage}
\begin{minipage}[b]{.5\textwidth}
\begin{center}
\includegraphics[angle=0,width=7.0cm]{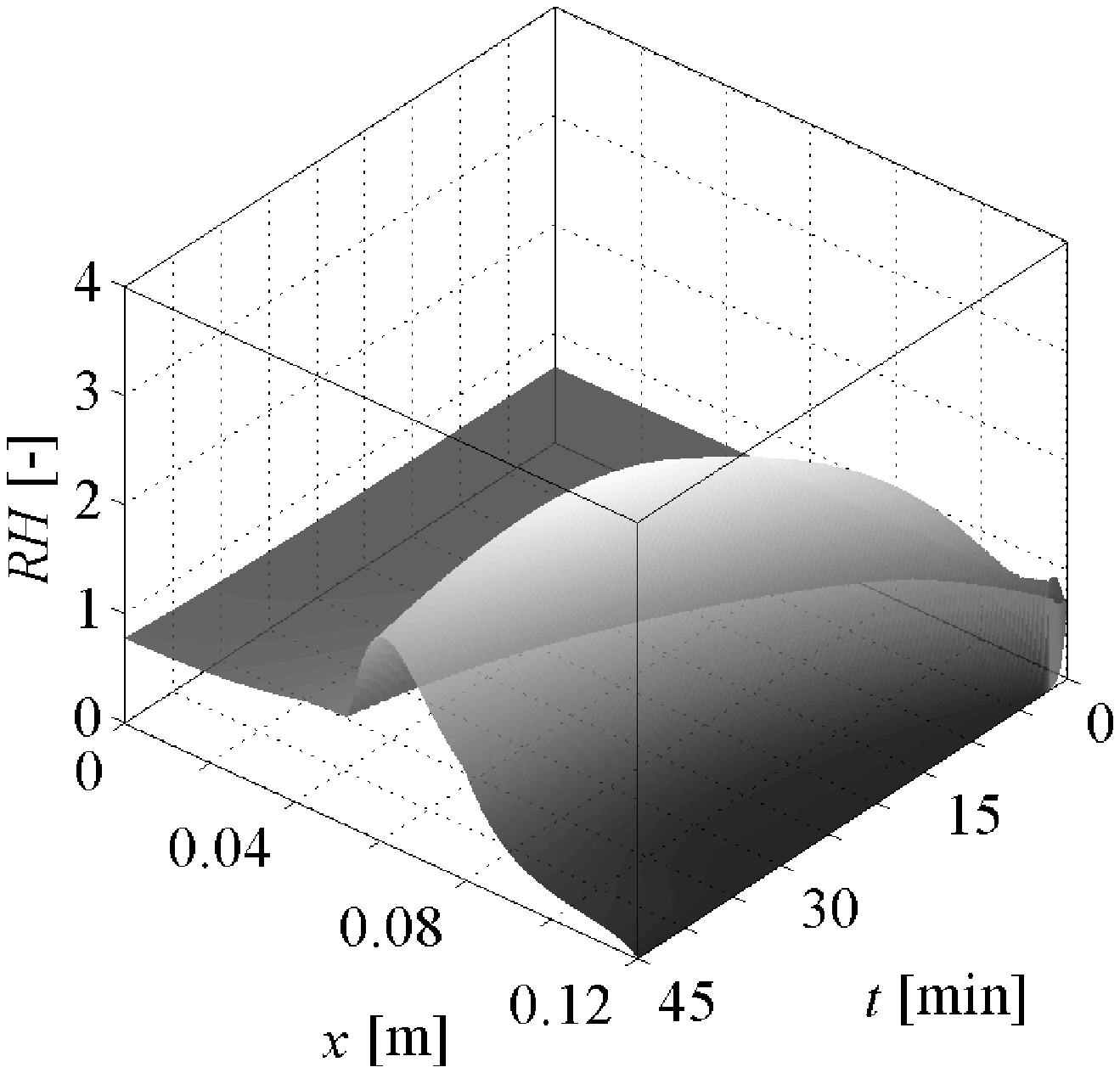}
\end{center}
\end{minipage}
\caption{Spatial and time distribution of temperature,
pore pressure and relative humidity for the analyzed wall
exposed to: HC fire (left) and PM fire (right)}\label{xt_Results}
\end{figure}

The results obtained for two different times of fire exposure are
given in Fig.~\ref{1D_Results}. It is obvious that the distribution
of the thermo-hygral quantities in concrete at high temperatures is
dependent on the type of fire scenario used for simulation. For
the nominal fires, where the fire temperature is a monotonically
increasing function of time (in our case the ISO fire and HC fire
curve), the temperature distribution across the analyzed wall
increases with time for the whole fire exposure. Contrary to this
for the PM fire, the temperature within the structure follows not
only the heating phase of a fire but also the cooling period, which
may lead to the decrease of temperature in some parts of a
structure, see Fig.~\ref{1D_Results} and Fig.~\ref{xt_Results}. This
holds also for
 the peak values of other thermo-hygral quantities, see Fig.~\ref{1D_Results}
 and Fig.~\ref{xt_Results}.

It is obvious, that the usage of the PM fire curve (or other natural
fire scenario) provides a better prediction of a real behavior of a
structure exposed to fire compared to the results obtained by
assuming some of the nominal fires (e.g. ISO fire or HC fire curve).
On the other hand, the nominal fire curves are still widely used due
to their simplicity and also due to the fact that, in most cases,
they lead to conservative results.

As shown in Figs.~\ref{1D_Results}~and~\ref{xt_Results},
the rapid heating of a structure and the related moisture
migration induce the increase of pore pressure near the heated surface.
In some cases, this pore pressure build up may lead to the spalling
of concrete surface layer, and hence, to the eventual collapse
of a structure, see~\cite{Benes2011}.


\bigskip

\paragraph{Acknowledgement.}
This outcome has been achieved with the financial support of the
Ministry of Education, Youth and Sports of the Czech Republic,
project No.~1M0579, within activities of the CIDEAS research centre.
Additional support from the grant 201/09/1544 provided by the Czech
Science Foundation and the grant SGS11/001/OHK1/1T/11 provided by
the Grant Agency of the Czech Technical University in Prague is
greatly acknowledged.

%
%
%

\end{document}